\documentclass[12pt]{article}
\usepackage[top=1.in,bottom=1.in,left=1.in,right=1.in]{geometry}
\usepackage[dvips]{graphicx}
\pdfoutput=1
\usepackage{color}
\usepackage[x11names,rgb]{xcolor}

\usepackage{tikz}
\usetikzlibrary{arrows,snakes,backgrounds}
\usepackage{setspace}
\usepackage{amsfonts}
\usepackage{amsmath}
\usepackage{amssymb}
\usepackage{amsthm}
\usepackage{enumerate}
\usepackage{graphicx}
\usepackage{float}
\usepackage[font=small,labelfont=bf]{caption}

\usepackage{subcaption}
\usepackage{float}
\usepackage[normalem]{ulem}
\usepackage[pdftex]{hyperref}
\usepackage[authoryear]{natbib}
\bibpunct{(}{)}{;}{a}{,}{,}
\hypersetup{colorlinks,%
	citecolor=blue,%
	filecolor=blue,%
	linkcolor=blue,%
	urlcolor=blue,%
}

\newtheorem{proposition}{{\bf \sc Proposition}}

\theoremstyle{remark}

\usepackage{authblk}

\newenvironment{revs}{\begin{color}{black} \ignorespaces} 
	{\end{color}}

\begin{document}

	\title{Marriage through friends}
	
	\author[a]{Ugo Bolletta}
	\author[b]{Luca P. Merlino\thanks{Corresponding author: LucaPaolo.Merlino@uantwerpen.be}}
	{
		\affil[a]{Universit\'{e} Paris-Saclay, RITM, Sceaux, France}
		\affil[b]{Department of Economics University of Antwerp, Antwerp, Belgium.}
	}
	
	\maketitle
	\begin{abstract}
		In this paper, we propose a model of the marriage market in which individuals meet potential partners either directly or through their friends. When socialization is exogenous, a higher arrival rate of direct meetings also implies more meetings through friends. When individuals decide how much to invest in socialization, meetings through friends are first increasing and then decreasing in the arrival rate of direct offers. Hence, our model \begin{revs}helps in rationalizing\end{revs} the negative correlation between the advent of online dating and the decrease of marriages through friends observed in the US over the past decades.
	\end{abstract}

	\section{Introduction}
	
	Since the seminal work of \cite{becker1973}, a great deal of attention has been devoted in economics (and sociology) to understand the determinants of \begin{revs}marriage, both in terms the degree of assortative matching in the marriage market, i.e., who marries with whom, and marriage rates \citep{stevenson2007,browning2014economics,schwartz2013trends}. However, the family is not a static institution, and marriage rates have considerably declined in the last decades \citep{stevenson2007}. While there are many factors explaining this trend, it is at least in part due to the difficulty to find a suitable partner \citep{gould2003}.
		
		As \cite{schwartz2005trends} point out,\end{revs} the search for a partner has been often conceptualized in a similar way as the search for a job: ``(j)ust as people have a reservation wage below
	which they will not take a job, they may also have a set of minimum qualifications they seek in a partner below which they will not form a match'' (\textit{ibid.}, p. 453). However, relative little attention has been devoted to \textit{how} people search for their partners, and the implications this has for \begin{revs}marriage rates across time (as the marriage market conditions change) and for different socio-demographic groups (who have different marriage market prospects).\end{revs}
	
	Indeed, while traditionally couples met mostly through common friends and acquaintances, the past two decades \begin{revs}have\end{revs} seen a major shift towards online dating platforms, which in the US has become the most popular way couples meet since 2013 \citep{rosenfeld2019}. As we show below, these results hold for both males and females, as well as for college and non-college educated people, with the difference that people with a college education started sooner to rely less on friends and relatives to find a partner. Additionally, also the total number of couples who met through friends has decreased in the last two decades. \begin{revs}While there is evidence that online dating increases marriage rates \citep{bellou2015}, its effect on how couples meet and how it depends on the marriage market conditions is not well understood.\end{revs}
	
	The aim of this paper is to develop a formal model of the heterosexual marriage market that can \begin{revs}help rationalizing\end{revs} these changes. In the model, singles meet either through their friends, or directly (e.g., by casual meetings or online dating). We assume that all individuals are initially married, but they face a risk of divorce. In order to reduce the probability of remaining single, they can invest in socialization with individuals of the same gender. This investment is determined trading off the cost of an additional unit of socialization with the additional meeting opportunities it brings. Indeed, we assume that, when married individuals meet a potential partner of the same or lower level of education as their current spouse, they introduce her to their single friends.\footnote{\begin{revs}As we discuss in Section \ref{model}, assuming everybody is initially married simplifies our analysis, but it does not qualitatively change the main predictions of our model. Indeed, singles would invest more in friendship than married people, as the first would rely more than the latter on their friends to find a partner. However, the predictions regarding how marriage rates through friends change as online dating becomes more prevalent would be the same.\end{revs}}
	
	We model friendship formation as a two-sided relationship among individuals of the same gender.\footnote{We focus on friends of the same gender as we want to capture strategic considerations in investments in friendship in a non-repeated setting. Indeed, friendship across genders can be strategic only if the favor of introducing a potential partner can be reciprocated in the future, which requires to model interactions that are repeated in time. \begin{revs}Furthermore, this is consistent with empirical findings that same-gendered peers have a stronger relationship and are more important for a variety of outcomes, at least among younger individuals \citep{merlino2019jole,soetevent}.\end{revs}} In particular, we assume the probability that two individuals become friends is increasing in the investment the two exert in socialization, and decreasing in the total socialization in the population, which accounts for congestion effects \citep{cabrales2011}.
	
	To illustrate the main mechanism present in our model, we first analyze the predictions of a model with homogeneous gains from marriage, meaning that marriage with \begin{revs}any\end{revs} potential partner \begin{revs}yields\end{revs} the same value. In this setting, we show that, if the network were taken as exogenous, the model would predict that the investment in socialization and, as a consequence, the number of couples who meet through friends, is increasing in the arrival rate of \begin{revs}partners through direct meetings\end{revs}. The intuition is that, as there more \begin{revs}partners meet through direct meetings\end{revs}, for example because of the appearance of online dating platforms, there are also more meetings through friends \begin{revs}because people who are already married meet more potential partners directly and pass them to their friends.\end{revs}
	
	However, when the network is endogenous, socialization is non-monotonic in the direct arrival rate of \begin{revs}partners through direct meetings\end{revs}. Indeed, when the arrival rate is initially low and the probability of remaining single is high, an increase in direct meetings induces people to invest more in friendship in order to take advantage of the availability of meetings through friends. However, when the arrival rate is sufficiently high, there already are many direct meetings, so individuals socialize less as the direct arrival rate further increases. Hence, as online dating becomes more prevalent, less \begin{revs}realized couples meet through friends.\end{revs}
	
	Hence, the introduction of endogenous effort in searching for partners via friends \begin{revs}
		helps rationalizing\end{revs} the empirical patterns that we document for US couples in the last decades as direct meetings via online dating apps became prevalent.
	
	We then introduce heterogeneity in the level of education, assuming that a marriage with a more educated partner entails a larger utility. \begin{revs}This is consistent with the findings of \cite{guner}, who find positive and increasing level of assortative mating in education in the last decades (see Figure 1, \textit{ibid.}). Furthermore, we assume that online meetings are more assortative than meetings through friends \citep{skopek}.\end{revs} We find that the same qualitative results hold with respect to how online dating affects marriages through friends. In particular, both educated and uneducated individuals first increase and then decrease their investment in the social network as the direct arrival rate of potential partners increases. The latter effect results in a decrease in marriages through friends when the direct arrival rate of potential partners is sufficiently high. Additionally, as empirically the decrease is more pronounced for educated individuals, our model suggests that the arrival rate of direct potential partners is higher for educated individuals, in line with the idea that they are more desirable partners.
	
	This paper naturally relates to the literature of network formation. As in \cite{jackson96}, in our \begin{revs}model\end{revs} mutual consent is needed in order for a link to be established. Contrarily to most models, players cannot decide with whom to interact, but only a generic socialization effort, as in \cite{cabrales2011} and \cite{currarini2009}. As a result, we can solve the game using the notion of Nash equilibrium, as in models with one-sided network formation \citep{balagoyal,KM,KM2,KM3}.
	
	The literature of marriage formation, initiated by \cite{becker1973}, typically assumes there is only a technology to directly meet potential partners. The functional form of this matching function has been the object of a lively literature. See \cite{choo2006} and \cite{chiappori2017} for some recent contributions. Peer effects in meeting rates in the marriage market have been recently studied in \cite{siow2021}. However, they do not explicitly study the role of marriages through friends. In the context of job search, this has been studied, for example, by \cite{calvo2004}, \cite{calvo2004effects}, \cite{galeotti2014}, \cite{merlino2014,merlino2019} and \cite{galenianos2021}. Following the modelling approach of \cite{galeotti2014}, this paper is the first to include these considerations in a model of marriage formation.
	
	The paper proceeds as follows. Section \ref{facts} describes some stylized facts on how people searched for their partners in the last decades in the US. Section \ref{model} introduces the model, \begin{revs}while Section \ref{sec:rates} derives the meeting rates through friends.\end{revs} Section \ref{homogenous} presents the results for the model with homogeneous individuals. Section \ref{heterogeneous} presents the results for the model with two levels of education. Section \ref{conclusions} concludes. All proofs are in Appendix.

	\section{Some stylized facts}\label{facts}
	
	We use the ``How Couples Meet and Stay Together'' HCMST 2017 public dataset \citep{dataset2}, which is a survey featuring 3510 respondents and including questions about how couples have met. In particular there is a set of questions that covers the use of on-line dating apps, such as \textit{Tinder} and \textit{Grindr}. We aim at using this data in order to match the implications derived with the model. 
	
	We restrict our dataset to couples that have met after 1992, in order to have better consistency with the adoption of on-line dating apps. This restriction reduces the sample size to 2158 observations. Since the model focuses on a network of acquaintances as intermediaries for marriage proposals, we define ``meetings through friends'' as those where the respondent declared to have met his/her partner through friends, family, coworkers and neighbors. There are 952 individuals that have met their partner through acquaintances, 396 through the internet and 1158 through the remaining channels.\footnote{In the residual category, there are meetings at church or religious activity, non-religious voluntary associations, military service, primary of secondary school, college, customer-client relationship, bar or restaurant, private party, public space, vacation, business trip, blind date and non-online dating services.} There are a number of individuals who declared multiple meeting channels, which explains the fact that we have 2516 declared meetings as opposed to only 2158 respondents. 
	
	In Figure \ref{fig:total}, we report the trend of both the matching rates and the total number of matches for the three different categories for couples who met between 1992 and 2017. Consistently with the findings of \cite{rosenfeld2019}, we find that meetings through acquaintances decrease over time for everyone, counterbalanced by as steady rise in online meetings. Thus, the other categories remain substantially unchanged over time. Overall, there is some substitution between meeting through acquaintances and online as time went by and the technology became more mature. In absolute terms, the only channel that experienced a decrease in the number of matches over time is matches through friends.
	
	\begin{figure}
		\centering
		\includegraphics[width=0.48\textwidth]{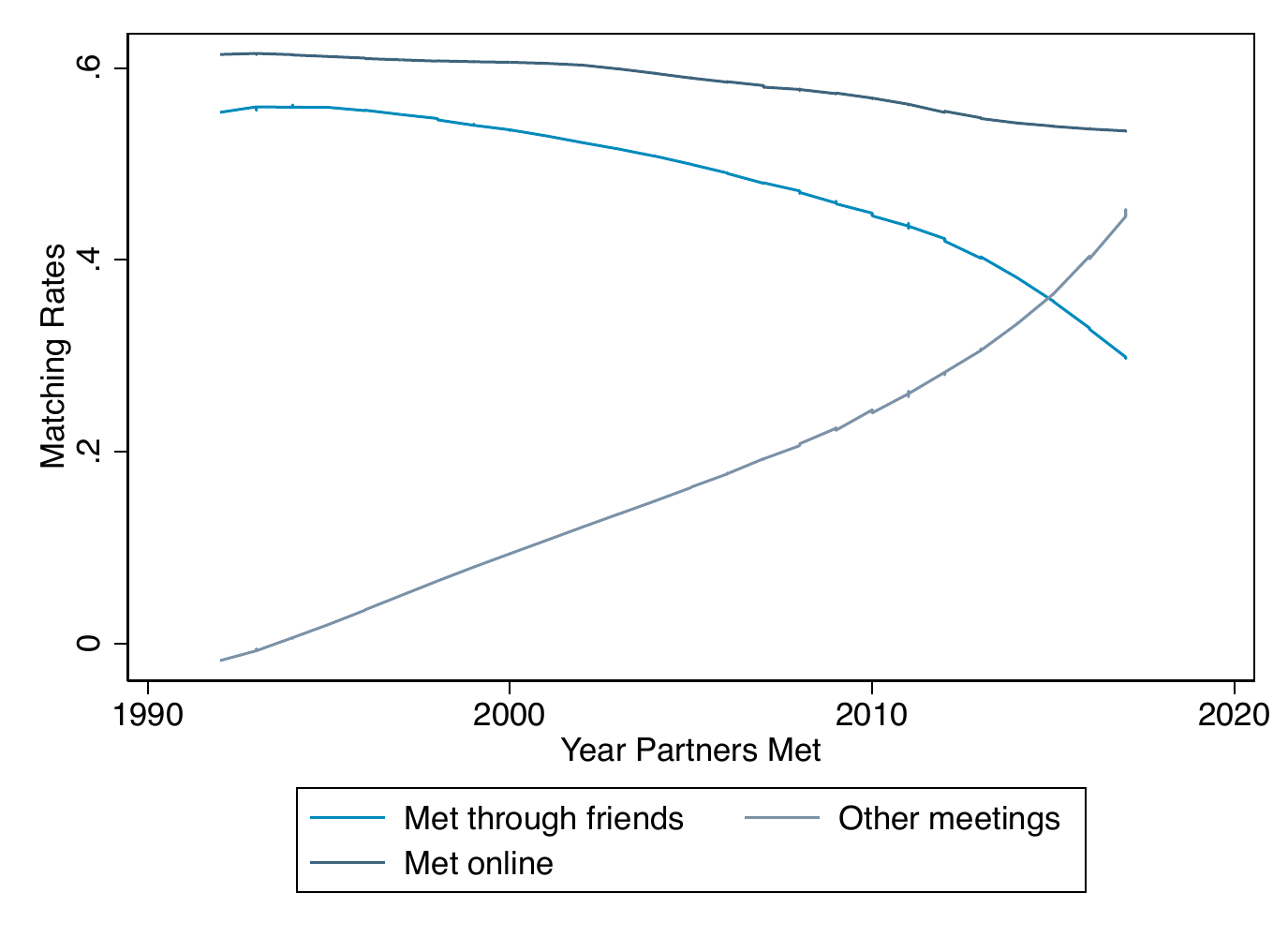}
		\includegraphics[width=0.48\textwidth]{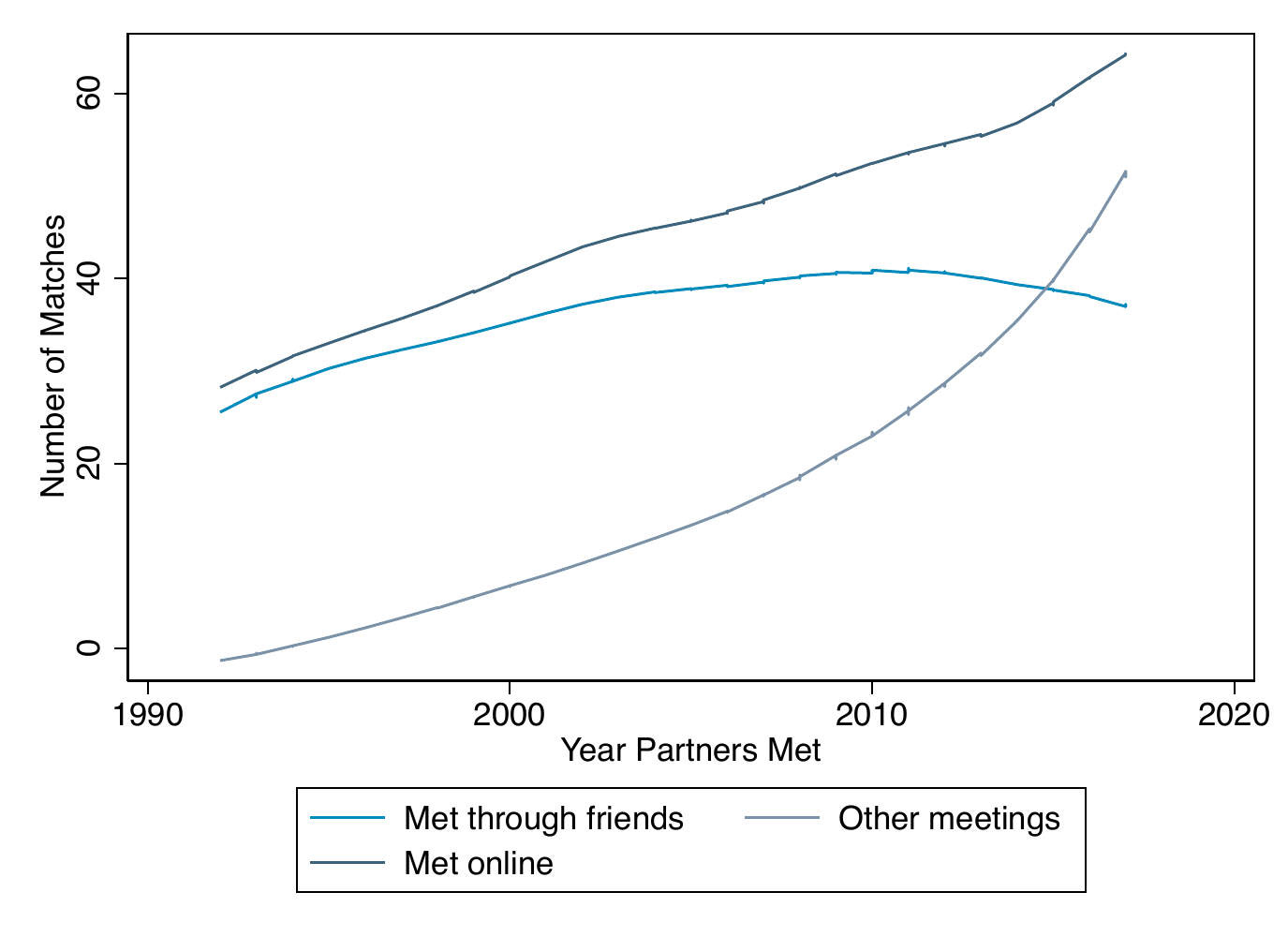}
		\caption{Trends of the matching rates and the total number of meetings occurred through different channels between 1992 and 2017 from unweighted Lowess regression with bandwidth $0.8$. Percentages do not add to $100\%$ because the categories are not mutually exclusive; more than one category can apply. Source: HCMST 2017.}
		\label{fig:total}
	\end{figure}
	
	We then look at trends differentiating by gender (Figure \ref{fig:bundlegen}) and by education, differentiating between people who have or not a college degree (Figure \ref{fig:bundlecol}). Looking at gender differences in Figure \ref{fig:bundlegen}, the trends are quite similar for both males and females, but there is a constant gap between meetings through acquaintances. This gap remains basically constant over time. Nevertheless, the other two categories do not exhibit a compensating difference. 
	
	\begin{figure}
		\centering
		\includegraphics[width=0.7\textwidth]{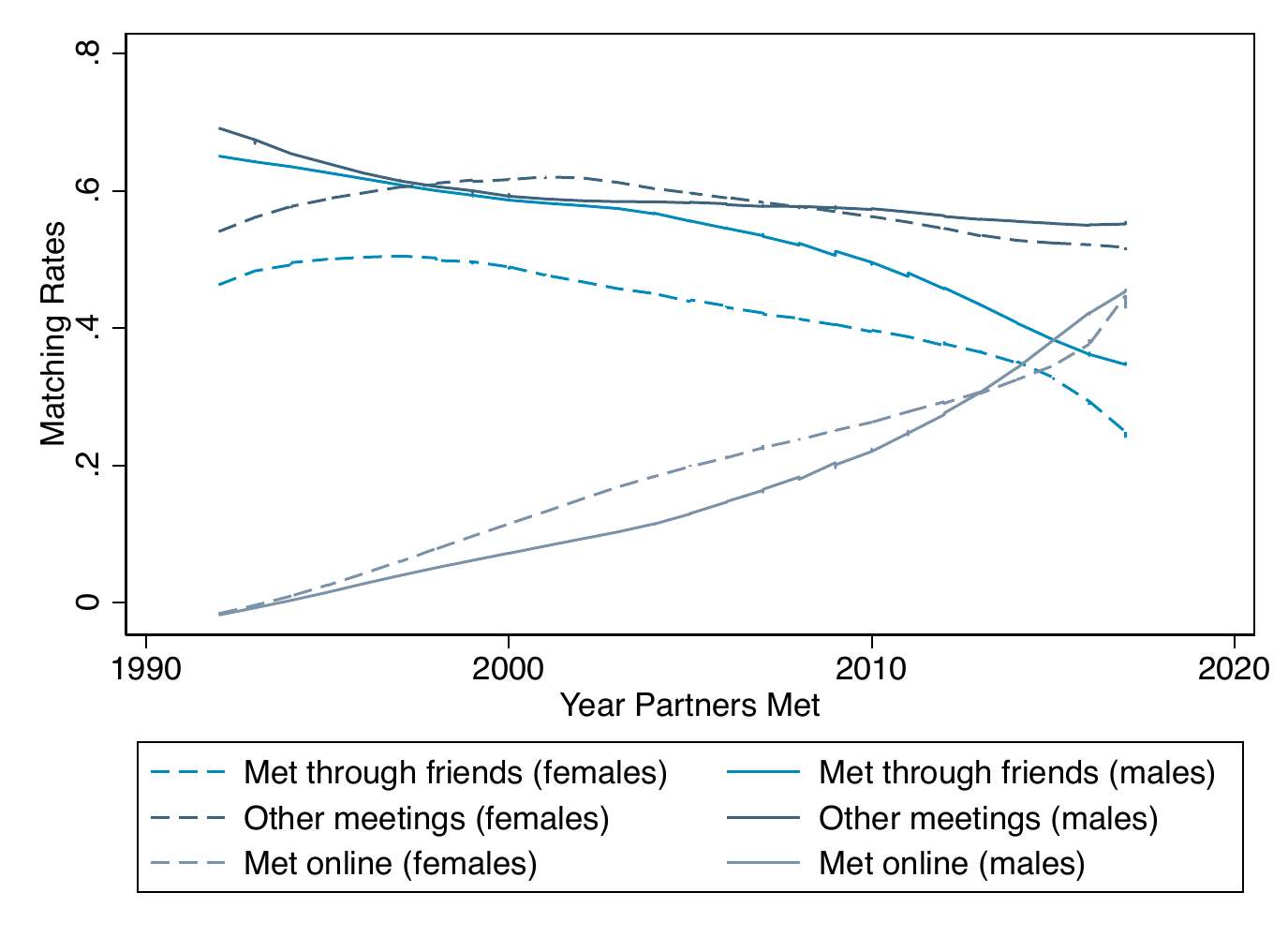}
		\caption{Trends by gender of the matching rates of meetings occurred through different channels between 1992 and 2017 from unweighted Lowess regression with bandwidth $0.8$. Percentages do not add to $100\%$ because the categories are not mutually exclusive; more than one category can apply. Source: HCMST 2017.}
		\label{fig:bundlegen}
	\end{figure}
	
	Looking at education, Figure \ref{fig:bundlecol} reveals that there is instead a strong difference between individuals with a college degree and those without. In fact, since the early 2000s, more educated individuals relied much more on the online dating apps than less educated individuals, and the gap has been growing ever since. The substitution towards online meetings has been to the detriment both of acquaintances and other meeting channels, although there is a more pronounced decrease of meetings through acquaintances. Furthermore, when looking at the total number of matches for each category, Figure \ref{fig:totalcol} reveals that, again, meetings through friends is the only category that experienced a decline over time, and the more so for more educated individuals.
	
	\begin{figure}
		\centering
		\includegraphics[width=0.7\textwidth]{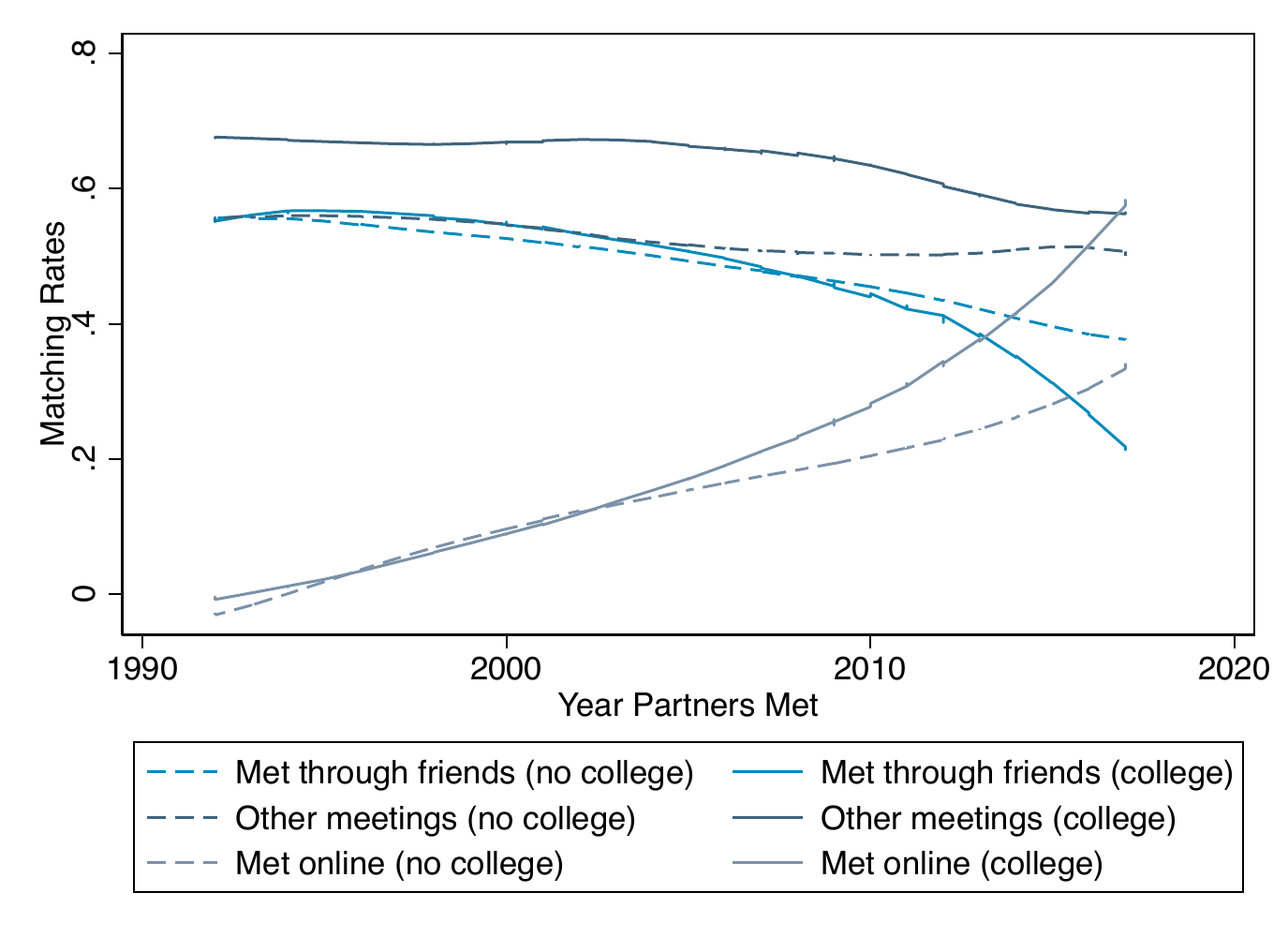}
		\caption{Trends by education (at least college education vs. not) of the matching rates of meetings occurred through different channels between 1992 and 2017 from unweighted Lowess regression with bandwidth $0.8$. Percentages do not add to $100\%$ because the categories are not mutually exclusive; more than one category can apply. Source: HCMST 2017.}
		\label{fig:bundlecol}
	\end{figure}
	
	\begin{figure}
		\centering
		\includegraphics[width=0.7\textwidth]{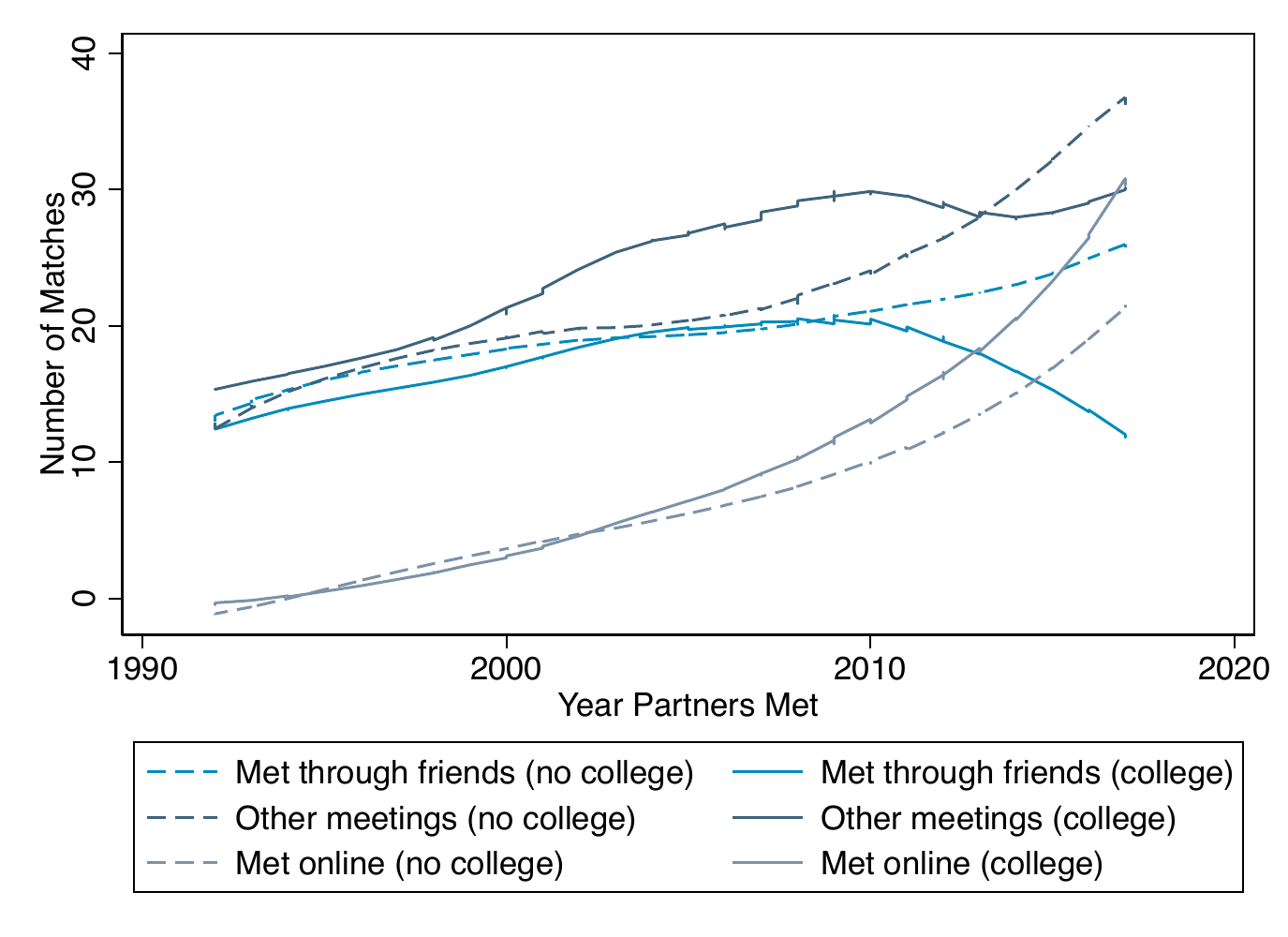}
		\caption{Trends by education (at least college education vs. not) of the total number of meetings occurred through different channels between 1992 and 2017 from unweighted Lowess regression with bandwidth $0.8$. Source: HCMST 2017.}
		\label{fig:totalcol}
	\end{figure}
	
	In the remainder of the paper, we show that 
	\begin{revs}
		taking into account individuals' choice on how much to rely on their friends in order to look for a partner helps us rationalizing the negative correlation between direct matches (either online or through other channels) and the number of matches through friends.\end{revs}\footnote{\begin{revs}
			For instance, \cite{rosenfeld2012} document that the increase in on-line meetings is accompanied by a decrease in meetings through friends. However, the authors mention that it is unclear whether overall on-line meetings could constitute a complement to real-life meetings. Thus, we contribute by showing how the endogenous effort in socialization can help rationalizing\end{revs} \begin{revs}that the substitution effect dominates. While also other factors might explain the displacement of meetings through friends, such as stigma around on-line dating, our model suggests that strategic considerations in socialization may be an important factor.\end{revs}}
	
	\section{The model}\label{model}
	
	We develop a simple model of an individual's decision to invest in friendship to meet potential partners through friends. \begin{revs}Individuals establish costly connections to increase the probability to find a potential partner through friends in case they become single. Since this event depends on marriage market conditions, the model allows us to derive predictions on how the investment in socialization and the proportion of marriages through friends change in the marriage market conditions, as captured by the arrival rate of direct meeting opportunities and the probability of divorce.\end{revs}
	
	Our model has three main parts: marriage formation through direct meetings, meetings through friends and network formation. We shall now explain in detail \begin{revs}each of these parts\end{revs}.
	
	\subsection{Marriage formation}
	
	The population consists of two groups of equal size $n$, males and females, respectively denoted by $M$ and $F$; hence, the total population is of size $2n$. Individuals can have a high or a low level of education. We denote by $m_H\in M_H$ ($m_L\in M_L$) \begin{revs}a high (low) educated male\end{revs}, with $M_H\subseteq M$ ($M_L\subseteq M$) being the set of high (low) educated men. Similarly, we denote by $f_H\in F_H$ ($f_L\in F_L$) \begin{revs}a high (low) educated female\end{revs}, with $F_H \subseteq F$ ($F_L \subseteq F$) being the set of high (low) educated females. The total proportion of high (low) educated individuals in the population is denoted with $h$ ($l$), that we assume to be equally distributed among males and females. We assume $h=1-l\in (0,1)$. \footnote{\begin{revs}We assume that men and women are completely symmetric. This choice is motivated by the findings in Figure \ref{fig:bundlegen}, where we show that there are not substantial gender differences when it comes to the way individuals meet their partner. This is probably due tot the fact that we observe only realized couples and not how each partner searched for a partner before finding one. A\end{revs} \begin{revs}more comprehensive analysis of heterogeneity in how search behavior differs across genders would require information on this search behavior, which does not exist. For this reason we focus our analysis considering heterogeneity only in education levels.\end{revs}}
	
	Individuals are either single or married with someone of the opposite gender.\footnote{To keep the model symmetric, we focus on heterosexual relationships. Homosexual relationships would be modeled in a similar fashion as in \cite{galeotti2014}.} We normalize the utility of being single to $0$. The gains of marriage instead depend on the partner's education. In particular, being married with a low educated individual gives \begin{revs}a utility of\end{revs} $1$, while the gains from marriage with a high skilled individual are $Y> 1$. \begin{revs}In the model of \cite{becker1973}, if individuals who are homogeneous but in their education level decide how much time to spend working and in home production, our assumption is equivalent to assuming that, \textit{(i)}, more educated individuals earn higher wages, and, \textit{(ii)}, partners' times spent in home production are gross complements. This could be due to agreement in parenting practices \cite{schwartz2013trends} and is consistent with the observation of a positive (or increasing) level of assortative mating in education in the last decades \citep{guner}.\end{revs}
	
	We assume that all $m\in M$ and all $f\in F$, \textit{(i)}, are initially married, \textit{(ii)}, with someone of the same type. This makes one's investment in the network depend only on own type. \begin{revs}As it will become clear below, these modelling assumptions introduce in a stylized way a feedback effect of the marriage market conditions on the socialization investment. If there were individuals who are initially single, their investment in socialization would be larger than for married individuals. The same would hold for worse matched individuals. However, these modifications to our model would not qualitatively affect the comparative statics with respect to marriage market conditions. Hence, our assumptions are not essential for the main results. However, they significantly simplify the notation and the analysis\end{revs}. 
	
	In the first period, some couples divorce with probability $d\in (0,1)$. Let us denote the set of individuals of education $e$ and gender $r$ who became single and who are dating (i.e., looking for a partner), by $\mathcal{D}_{r,e}$, for $r=M,F$ and $e=H,L$. \begin{revs}For simplicity, we assume that the divorce rate is exogenous and homogeneous across the population. This assumption simplifies the analysis. In our model, endogenous divorces would imply a higher risk of divorce for low educated individuals. This is because couples composed of low educated individuals have an incentive to divorce to find a more desirable partner, while couples composed of highly educated individuals would never divorce. This relationship between education and risk of divorce is indeed negative in the US in reality, and it has become increasingly negative in several other countries \citep{harkonen2006}. We discuss below the effect of assuming a higher divorce rates for less educated individuals.\end{revs}
	
	Subsequently, each individual meets a person of the opposite gender and the same skill level with an exogenous probability $a\in(0,1)$. \begin{revs}The complementary probability $(1-a)$ describes the event in which an individual does not meet any other person of the opposite gender directly. Let us denote\end{revs} the set of individuals of education $e$ and gender $r$ who directly met a potential partner by $\mathcal{A}_{r,e}$, for $r=M,F$ and $e=H,L$.
	
	Under this protocol, nobody \begin{revs}meets more than one potential partner\end{revs}. The set $\mathcal{U}_{r,e}=\mathcal{D}_{r,e}\cap R_e \setminus \mathcal{A}_{r,e}$ is the set of individuals of gender $r=m,f$ (hence belonging to $R=M,F$) and education $e=h,l$ who divorced and did not directly meet another potential partner. Note that $\vert\mathcal{U}_{r,e}\vert=d(1-a)e n$, for $e=h,l$. The set $\mathcal{O}_{r,e}=\mathcal{A}_{r,e}\cap R_e \setminus \mathcal{D}_{r,e}$ contains individuals of gender $r=m,f$ (hence belonging to $R=M,F$) and education $e=h,l$ who did not divorce and meet a potential partner of the same \begin{revs}educational level\end{revs}. We say that an individual $i\in \mathcal{O}_{r,e}$ has a needless date. Note that there are $\vert\mathcal{U}_{r,e}\vert=a(1-d)e n$ needless dates, for $e=h,l$.
	
	\subsection{Meetings through friends}
	
	At this point, meetings through friends occur. In particular, \begin{revs}an individual can meet a potential partner through friends in two ways: through own friends who are married and through the friends of someone met directly who is married.
		
		The first channel includes any married individual who\end{revs} has a \textit{needless potential partner} to introduce to one of his/her friends. We denote by $\Psi_{i,j}$ the probability that an individual of type $i$ will get introduced to a potential partner of type $j$ by one of $i$'s friends. \begin{revs}Throughout the paper we assume that $\Psi_{h,l}=0$, that is, low skilled individuals do not introduce low skilled individuals to their high skilled friends. Consistently with the hypothesis of exogenous divorces, this implies that there is no option value in turning down a potential partner, so that all realized meetings translate into marriages. We think relaxing this assumption would be interesting in a repeated version of this model.
		
		Additionally, one can meet a potential partner also via someone who was met directly, but who is married and introduces him/her to one of her single friends. We denote by $\Upsilon_{i,j}$ the probability that an individual of type $i$ will get introduced to a potential partner of type $j$ who was introduced to him/her by someone met directly, but who was married.\end{revs}
	
	Upon realization of meetings, marriages are formed.
	
	\subsection{Network formation}
	
	Anticipating the risk of becoming single, individuals invest in a network of friends to increase the probability to find a partner in case this happens.
	
	In particular, males (females) are organized in an undirected friendship network $g\in G$, $G$ being the class of all possible network. All links $g_{ij}$ are such that $g_{ij}=\{0,1\}$ for all $i,j\in M$ ($i,j\in F$) with $i\neq j$. To simplify the analysis, we assume no links are formed across genders, and so $g_{ij}=\{0\}$ for all $i\in M$ and $j\in F$. We call the set of connections individual $m\in M$ ($f\in F$) has in network $g$ and that belong to a partition $V\subset M$ ($V\subset F$), $N_m(V)=\{j\in V\setminus m \vert g_{ij}=1\}$ ($N_f(V)=\{j\in V\setminus f \vert g_{ij}=1\}$). Let us define $\eta_i(V)=\vert N_i(V) \vert $ the number of connections individual $i$ has within the subset of neighbors that belong to $V$.
	
	\begin{revs}
		Connections are established as follows \citep{cabrales2011}. Each individual $i\in N$ chooses a socialization investment $s_i \geq 0$, which has a constant and homogeneous marginal cost equal to $c$. Denote the set of pure strategies of individual $i$ by $S_i$. A pure strategy profile is $\mathbf{s}=(s_1,...,s_n)\in s= \mathbb{R}^n_+$, and $\mathbf{s}_{-i}$ denotes the set of all strategies other than that of player $i\in N$. Denote by $y(\mathbf{s}_m)$ and $y(\mathbf{s}_f)$ the aggregate socialization investment for males and females respectively. Given a strategy profile $\textbf{s}$, a link between two arbitrary individuals $i,j\in R$ with $R=\{M,F\}$ is realized with probability\end{revs}
	\begin{equation*}
		Pr(g_{ij}=1 \vert \textbf{s})= \left\{ \begin{array}{ll}
			\min \left\{\frac{s_is_j}{y(\textbf{s}_r)},1\right\} &\quad \mbox{if} \;y_r(\textbf{s})>0,\\
			0 &\quad \mbox{otherwise.}
		\end{array}\right.
	\end{equation*}
	\begin{revs} Hence, a socialization profile $\mathbf{s}$ generates a multinomial random graph. Note also that, given our assumptions, the maximization problem that females and males of the same skill level face is the same. Hence, when workers of the same gender and skill level choose the same level of investment, the probability that an individual of gender $r$ and education level $e=h,l$ has a friend of the same gender and education $e'=h,l$ is\end{revs}
	\[
	\min\{\frac{s_e s_{e'}}{n(hs_h+(1-h)s_l)},1\}.
	\]
	
	\subsection{Timing}
	
	To sum up, the timing of the model, depicted in Figure \ref{fig:timing}, is as follows:
	\begin{enumerate}
		\item All individuals are married and choose a costly investment in the network $s$. The network is formed accordingly.
		\item Divorces occur with probability $d$ and direct meetings with potential partners occur with probability $a$.
		\item Potential partners who met a married individual are introduced by the date to a friend. If the friend is single, a marriage occurs. 
	\end{enumerate}
	
	\begin{figure}[htb!]
		\centering
		\includegraphics[width=\textwidth]{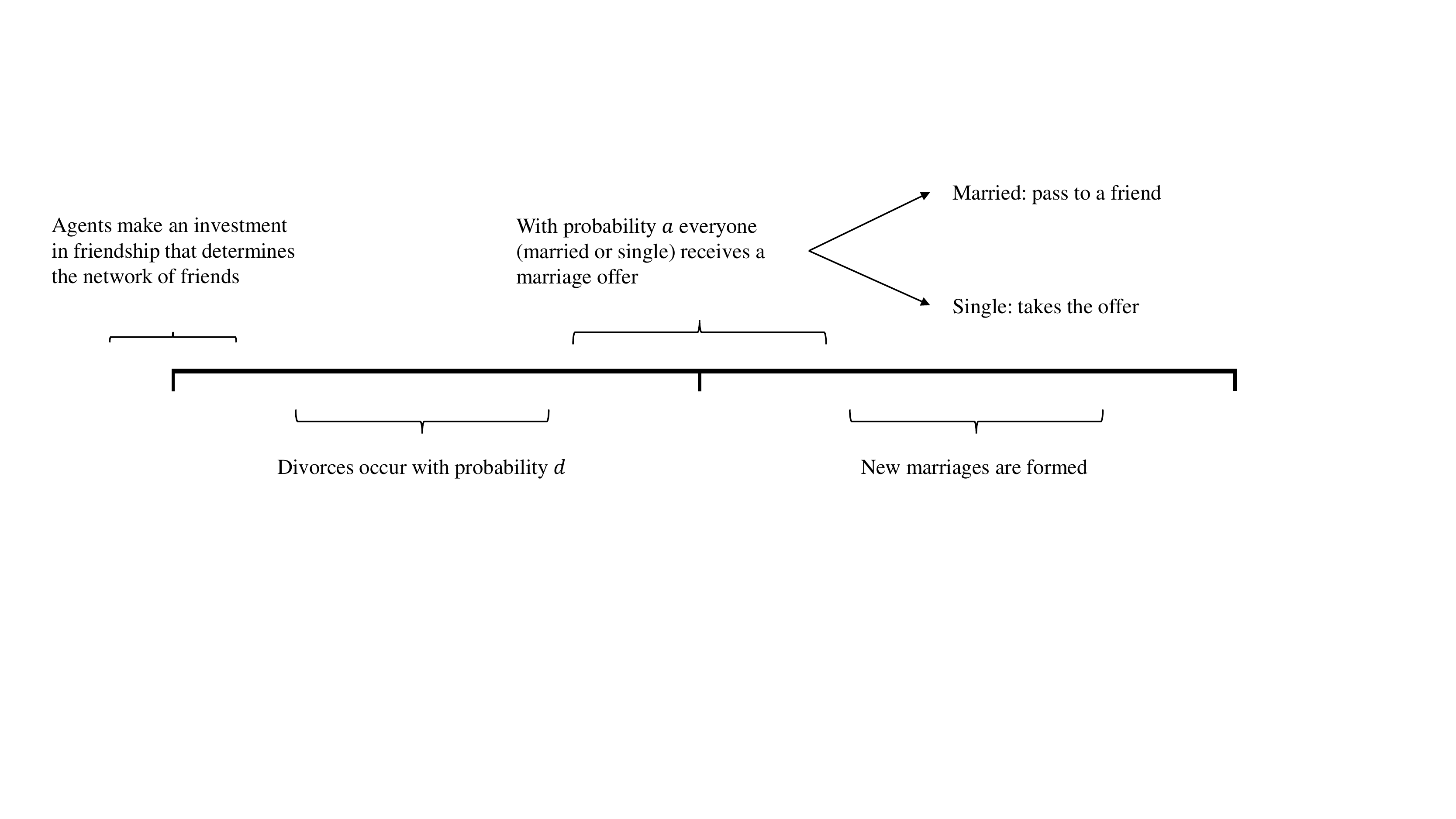}
		\caption{Timing of the model.}\label{fig:timing}
	\end{figure}
	
	\subsection{Utilities and equilibrium}
	
	For a strategy profile $\textbf{s}$, we denote by $\Psi_{i,r,e,e'}(\textbf{s})$ the probability that a single individual $i$ of gender $r=m,f$ and education level $e=h,l$ has access to at least one date with an individual of the opposite gender and education level $e'=h,l$ introduced by a (married) friend. \begin{revs}Additionally, we denote by $\Upsilon_{i,r,e,e'}(\textbf{s})$ the probability that a single individual $i$ of gender $r=m,f$ and education level $e=h,l$ has access to at least one date with an individual of the opposite gender and education level $e'=h,l$ who was introduced to him/her by someone met directly but who was married.
	\end{revs} Then, the expected utility of an individual $i$ of gender $r=m,f$ and education level $e=h,l$ is
	\begin{eqnarray}
		EU_{i}(\textbf{s})&=& (1-d)M_{e}+d^2 a M_{e}+ d a (1-d) [\Upsilon_{i,r,e,h}(\textbf{s}) Y + \Upsilon_{i,r,e,l}(\textbf{s}) ]+ \notag \\
		&+& d (1-a)[ \Psi_{i,r,e,h}(\textbf{s}) Y+(1-\Psi_{i,r,e,h}(\textbf{s}))\Psi_{i,r,e,l}(\textbf{s})]-c s_i,\label{eq:expectutility}
	\end{eqnarray} 
	where \begin{revs}the gains of marriages are\end{revs} $M_{e}=Y$ if $e=h$, and $1$ otherwise. The first term captures the event in which individual $i$ does not divorce, and is married with someone of the same skill level. In case of divorce, $i$ gets a direct date with probability $a$. In this case, if the date is single, they marry; otherwise, the dating opportunity is transferred to a single friend, who can be of either type. Finally, if $i$ \begin{revs}does not meet someone directly, $i$ can meet someone\end{revs} through one (of either type) of his/her friends. Note that the assumption that individuals of type $l$ \begin{revs}do not pass $l$ dates to their friends of type $h$ is equivalent to assuming $\Psi_{i,r,h,l}=0$ for all $i\in N$ and $r=m,f$.\end{revs}
	
	\begin{revs}The marriage rate through friends is the probability that an individual meets a partner through his friends, either because they are introduced by someone they have met directly but was married--$\Upsilon$,---or because they got introduced a potential partner via a friend---$\Psi$. Formally, denote the marriage rate through friends of individual $i$ of gender $r=m,f$ and education $e=h,l$ by $m_{i,r,e}(\textbf{s})$, which is then derived as follows:\end{revs}
	\begin{eqnarray*}
		m_{i,r,e}(\textbf{s})&=& \Upsilon_{i,r,e,h}(\textbf{s}) + \Upsilon_{i,r,e,l}(\textbf{s}) + \Psi_{i,r,e,h}(\textbf{s}) +\Psi_{i,r,e,l}(\textbf{s}).\label{marriage_rate}
	\end{eqnarray*}
	
	A pure strategy equilibrium in this context is a strategy profile $\textbf{s}$ such that the following holds: 
	\begin{equation*}
		EU_i(s_i,\textbf{s}_{-i})\geq EU_i(s'_i,\textbf{s}_{-i}) \forall s_i'\in S_i.
	\end{equation*}
	In a symmetric equilibrium, $\Psi_{i,r,e,e'}(\textbf{s})=\Psi_{e,e'}$ for all $i\in N$ with $e=h,l$ and $e'=h,l$. Hence, in a symmetric equilibrium, the expected utility \eqref{eq:expectutility} for the high types and the low types can be written as:
	\begin{eqnarray*}
		EU_{h}&=& (1-d)Y+d^2 a Y+ d a (1-d)[\Upsilon_{h,h} Y+ \Upsilon_{h,l}] +\notag \\
		&+& d (1-a)\Psi_{h,h} Y-c s , \label{eq:exputilh} \\
		EU_{l}&=&(1-d)+d^2a +d a (1-d) [\Upsilon_{l,h} Y + \Upsilon_{l,l}]+ \notag \\
		&+& d (1-a)[\Psi_{l,h} Y+(1-\Psi_{l,h}) \Psi_{l,l}] -c s . \label{eq:exputill}
	\end{eqnarray*}
	
	\section{Marriage rate through friends}\label{sec:rates}
	
	To derive the marriage rates through friends, we focus on networks generated by a symmetric profile of network investment, i.e., $s_{i,r,e}=s_e$ for all $i\in \mathcal{N}$ of gender $r=m,f$ and type $e=h,l$. This is because in a symmetric equilibrium all individuals of the same educational irrespective of their gender exert the same level of socialization investment. 
	\begin{revs} As in our model marrying a $h$ type induces larger gains than marrying a $l$ type, one marries a less educated individual only if no meeting with a highly educated individual occurred.\end{revs}
	
	First, we derive the probability $\Upsilon_{h,h}$ that an individual $i$ of type $h$ who received a dating opportunity directly with an individual $j$ of the same type $h$ but married, gets introduced to someone of type $h$, who is a friend of the original contact. This is the probability that $j$ has at least one friend of type $h$. As the probabilities that $j$ has one friend of type $h$ is
	\[
	p_{h,h}=\frac{s_h^2}{n[h s_h+ (1-h)s_l]},
	\]
	the probability that $j$ of type $e$ has at least a single $h$ friend as $n\rightarrow \infty$ is 
	\[
	\lim_{n\rightarrow \infty} \Upsilon_{h,h}=\lim_{n\rightarrow \infty} 1 - \left( 1 - \frac{s_h^2}{n(h s_h+ (1-h)s_l}\right)^{n d h}=1-e^{-\frac{d h s_h^2}{h s_h+ (1-h)s_l}}.
	\]
	The other probabilities $\Upsilon_{h,l}$, $\Upsilon_{l,h}$ and $\Upsilon_{l,l}$ are derived accordingly. Note that $\Upsilon_{h,l}=\Upsilon_{l,h}$.
	
	\begin{revs}
		We now derive the probability that an individual $i$ of type $h$ who has not met someone directly is introduced to someone of type $h$ by one of his/her friends. To do so, we first calculate the probability that $i$ does not meet any $h$ individual of the opposite gender from his/her $h$ friends, which we denote by $\phi_{h,h}(s)$ and which is\end{revs}
	\begin{equation*}
		\phi_{h,h}(s)= \left(1- p_{h,h} \times \frac{1-\left(1- p_{h,h}\right)^{d h n}\left(1- p_{h,l}\right)^{d (1-h) n}}{d \left(h (n-1) p_{h,h}+(1-h)n p_{l,l}\right)}\right)^{a (1-d) h n}.
	\end{equation*}
	Since $p_{h,h}=s_h^2/[n(hs_h+(1-h)s_l]$ and $p_{h,l}=s_h s_l/[n(h s_h+(1-h)s_l]$, in a large marriage market we obtain that the network matching rate among highly educated individuals is
	\begin{equation*}\label{psihh}
		\Psi_{h,h}(s)=1-\lim_{n \rightarrow \infty} \phi_{h,h}(s) = 1-e^{-\frac{a(1-d)h s_h}{d (h s_h+(1-h)s_l)}\left(1-e^{-d s_h}\right)}.
	\end{equation*}
	In a similar way, we derive the network meeting rate for low types to get introduced to a high type:
	\begin{eqnarray*}
		\Psi_{l,h}(s)&=&1-e^{-\frac{a(1-d) h s_l}{d (h s_h+(1-h)s_l)}\left(1-e^{-d s_h}\right)}, \label{psilh}
	\end{eqnarray*}
	
	To derive $\Psi_{l,l}$, we need to take into account the assumption that an individual with low education never introduces someone of the same type to a high skilled individual. Hence, there is less competition in the social network for these partners. In particular, we can write
	\begin{equation*}
		\phi_{l,l}(s)= \left(1- p_{l,l} \times \frac{1-\left(1- p_{l,l}\right)^{d (1-h) n}}{d (1-h) n p_{l,l}}\right)^{a (1-d) (1-h) n}.
	\end{equation*}
	Since $p_{l,l}=s_l^2/[n(hs_h+(1-h)s_l]$, in a large marriage market, we obtain that the network matching rate between individuals of low education is
	\begin{eqnarray*}
		\Psi_{l,l}(s)&=&1-e^{-\frac{a(1-d)(1-h)s_l}{d (h s_h+(1-h)s_l)}\left(1-e^{-d (1-h) s_l\frac{sl}{h s_h+(1-h) s_l} }\right)} . \label{psill}
	\end{eqnarray*}
	
	Finally, given our assumptions that an individual with low education never introduces someone of her/his type to a high skilled individual, \begin{revs}$\Psi_{h,l}(s)=0$.\end{revs}
	
	\section{Homogeneous gains from marriages}\label{homogenous}
	
	Before turning to the full model with heterogeneous skill levels, we first highlight the main predictions of our model of marriage through friends in a model where individuals differ only in their gender. We first consider the case where the friends network is given and we explicitly derive the network matching rate. We then derive the equilibrium network and provide comparative statics result. In this section, we remove reference to group $l$ as there is only one level of education and no confusion may arise.
	
	\subsection{Exogenous friendship network}
	
	Since $p=s/n$, in a large marriage market where $h=0$ we obtain that the network matching rates are
	\begin{eqnarray*}
		\Upsilon(s,a)&=&1-e^{-s d}, \label{upsi_homo}\\
		\Psi(s,a)&=&1-e^{-\frac{a(1-d)}{d}(1-e^{-s d})} \label{psi_homo}.
	\end{eqnarray*}
	The expected utility to an individual $i\in M\cup F$ is:
	\begin{equation*}\label{utility}
		EU_{i}(s_{i},\mathbf{s}_{-i})=(1-d)+d^2 a+da(1-d)\Upsilon(\mathbf{s},a)+d(1-a)\Psi _{i}(\mathbf{s},a)-cs_{i}.
	\end{equation*}
	When the investment in socialization is exogenous, we derive the following proposition.
	\begin{proposition}\label{prop:exogenous}
		Consider a large marriage market where $h=0$ and $s_{i}=s$ for all $i\in\mathcal{N}$. Then, the network matching rates $\Upsilon(s)$, $\Psi(s)$ and $m(s,a)$ are increasing in the socialization effort, $s$, and weakly increasing in the direct arrival rate of potential partners, $a$.
	\end{proposition}
	\begin{revs}
		Intuitively, as socialization increases, individuals meet more potential partners through their friends. This directly translates into more people finding a partner through their social network. An increase in the direct arrival rate of potential partners has the same effect, as there are more potential partners who can be introduced a friend.\end{revs}
	
	\subsection{Endogenous friendship network}
	
	\begin{revs}
		We now allow individuals to choose strategically how much to invest in socialization, thereby endogenizing the network.\end{revs}
	
	The following proposition characterizes interior equilibria.\footnote{Note that there always exists an equilibrium where individuals do not invest in the network, i.e., $s_i=0$ for all $i\in\mathcal{N}$. \begin{revs}If $c\geq ad(1-a)(1-d)$, this equilibrium is the only symmetric equilibrium. Otherwise, the best response dynamics after a perturbation around the equilibrium reveals that this equilibrium is unstable.\end{revs} 
	}
	\begin{proposition}\label{prop:NE} 
		Consider a large marriage market with $h=0$. An interior symmetric equilibrium $s^{\ast}$ exists if, and only if, $c<a d(1-a)(1-d)$, and $s^{\ast }$ is the unique solution to: 
		\begin{equation}\label{foc_homo}
			d(1-a)\left[\frac{a(1-d)}{d s^{\ast }}\left( 1-e^{-s^{\ast }d}\right) e^{-\frac{a(1-d)
				}{d}\left( 1-e^{-s^{\ast }d}\right) }\right]=c.
		\end{equation}
	\end{proposition}
	The proof of this result follows \cite{galeotti2014}. Indeed, while our model is different, the incentives to invest in socialization in the model with \begin{revs}homogeneous\end{revs} individuals are the same, as an individual's investment does not directly affect the probability to get introduced to a friend of a married person met in a direct meeting, $\Upsilon$.
	
	\begin{revs}
		In words, there exists a unique interior equilibrium, which we derive by equalizing an individual's marginal returns to socialization with its marginal costs. The marginal returns are the marginal increase in the friends matching rate, $\Psi$, multiplied by the likelihood an individual needs the network to find a partner, $d(1-a)$. 
		
		The following proposition derives some comparative statics of the interior equilibrium with respect to the arrival rate of direct meetings with potential partners.\end{revs}
	\begin{proposition}\label{prop:cs}
		Consider a large marriage market and suppose that $c<ad(1-a)(1-d)$. Then:
		\begin{itemize}
			\item[1.] For every $d\in(0,1)$, there exists a unique $\bar{a}(d)>0$ such that if $a\leq \bar{a}(d)$, then socialization increases in the arrival rate of direct meetings with potential partners $a$, otherwise it decreases in the arrival rate of direct meetings with potential partners.
			\item[2.] For every $d\in(0,1)$, if the arrival rate of direct meetings with potential partners, $a$, is low enough, then the network matching rates $\Upsilon$, $\Psi$ and $m$ are increasing in the arrival rate of direct meetings with potential partners, $a$, while if the arrival rate, $a$, is large enough, then they are decreasing in the arrival rate of direct meetings with potential partners, $a$.
		\end{itemize}
	\end{proposition}
	\begin{revs}
		The first part of Proposition \ref{prop:cs} shows that socialization is first increasing and then decreasing in the arrival rate of direct meetings. Hence, socialization is highest for intermediate values of this arrival rate. Intuitively, when the arrival rate is low, there is not much value in investing in connections because there are not many potential dates that can be passed on by one's friends. When the arrival rate is high, the value of investing in connections is also low, because there are many direct meetings taking place outside the friendship network. So, socialization is stronger in marriage markets where the direct arrival rate of meetings is moderate.\end{revs}
	
	The non-monotonic relationship between socialization and the arrival rate of direct meetings translates into a similar relationship between network marriage rates and arrival rate of direct meetings. \begin{revs}Indeed, when this is low, its increase triggers more socialization, so that more potential partners meet through friends, and more people find a partner in this way. When the arrival rate of direct meetings increases too much, a congestion effect emerge, which is due to the fact that the network becomes crowded with other potential partners. This reduces socialization, and, eventually, also the network marriage rate.\end{revs}
	
	Figures \ref{fig:socialiation_hom} and \ref{fig:rates_hom} depict how socialization investment and the matching rate through friends, $\Psi$, change when the direct arrival rate of potential partners, $a$, changes from $.3$ to $.7$ in a marriage market where we assume that $c = .005$, $d = .015$ and the value of marriage is $2$. In Figure \ref{fig:rates_hom}, we report only the meeting rate for partners introduced by a married friend, $\Psi$, as the network meeting rate of meetings through a married potential partner, $\Upsilon$, is strictly monotonic in socialization investment.
	
	In line with the results of Proposition \ref{prop:cs}, socialization is first increasing and then decreasing in the arrival rate of potential partners, as an increase in the direct arrival rate eventually reduces the incentives to invest in socialization. The effect of this reduction in the investment in socialization eventually translates into less marriages through friends as the arrival rate of direct meetings increase.
	
	\begin{figure}
		\centering
		\includegraphics[width=0.7\textwidth]{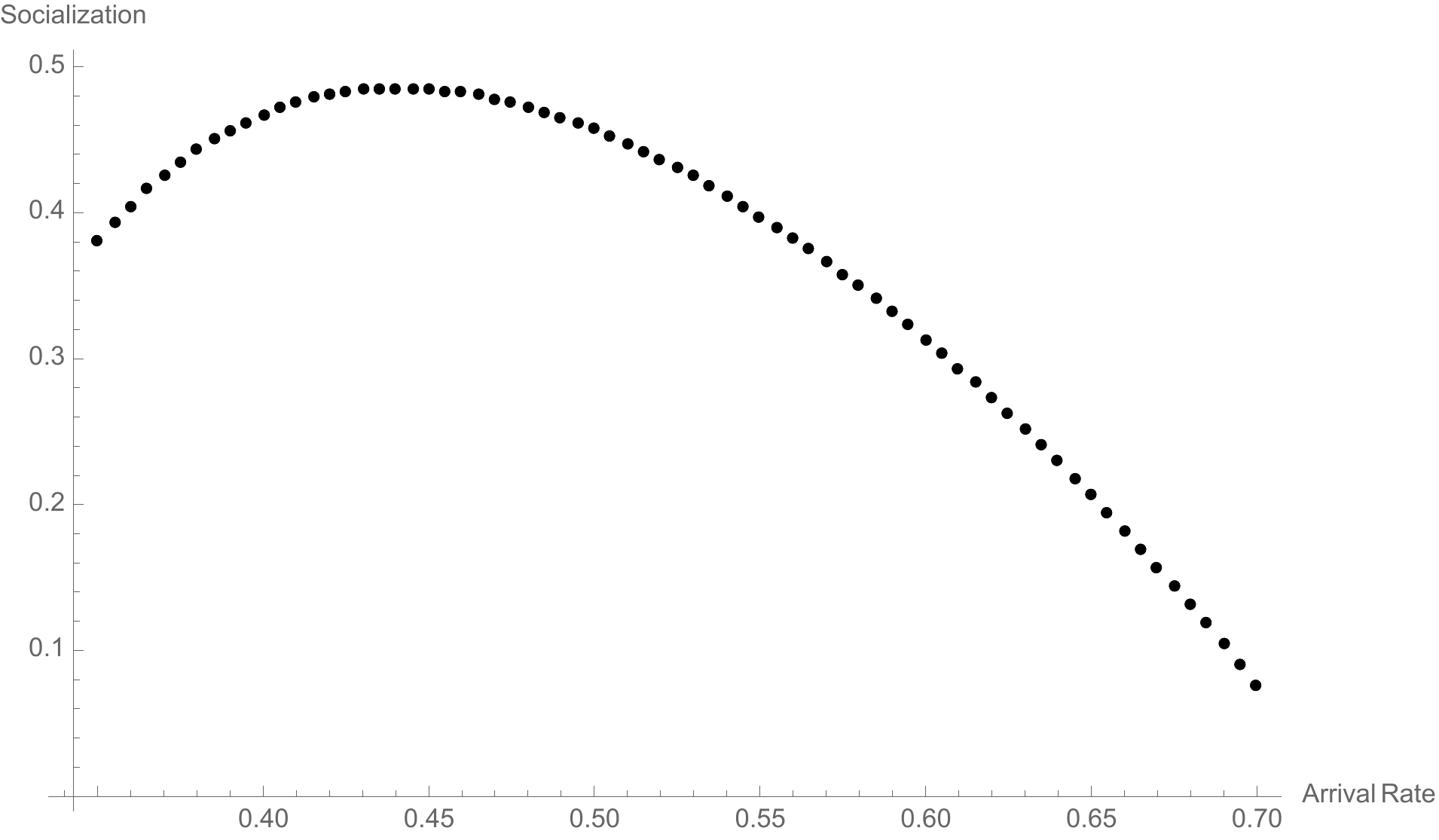}
		\caption{Equilibrium socialization investment in the model with homogeneous individuals ($h=1$) when $c = .005$, $d = .015$ and $Y = 2$.}
		\label{fig:socialiation_hom}
	\end{figure}

	\begin{figure}
		\centering
		\includegraphics[width=0.7\textwidth]{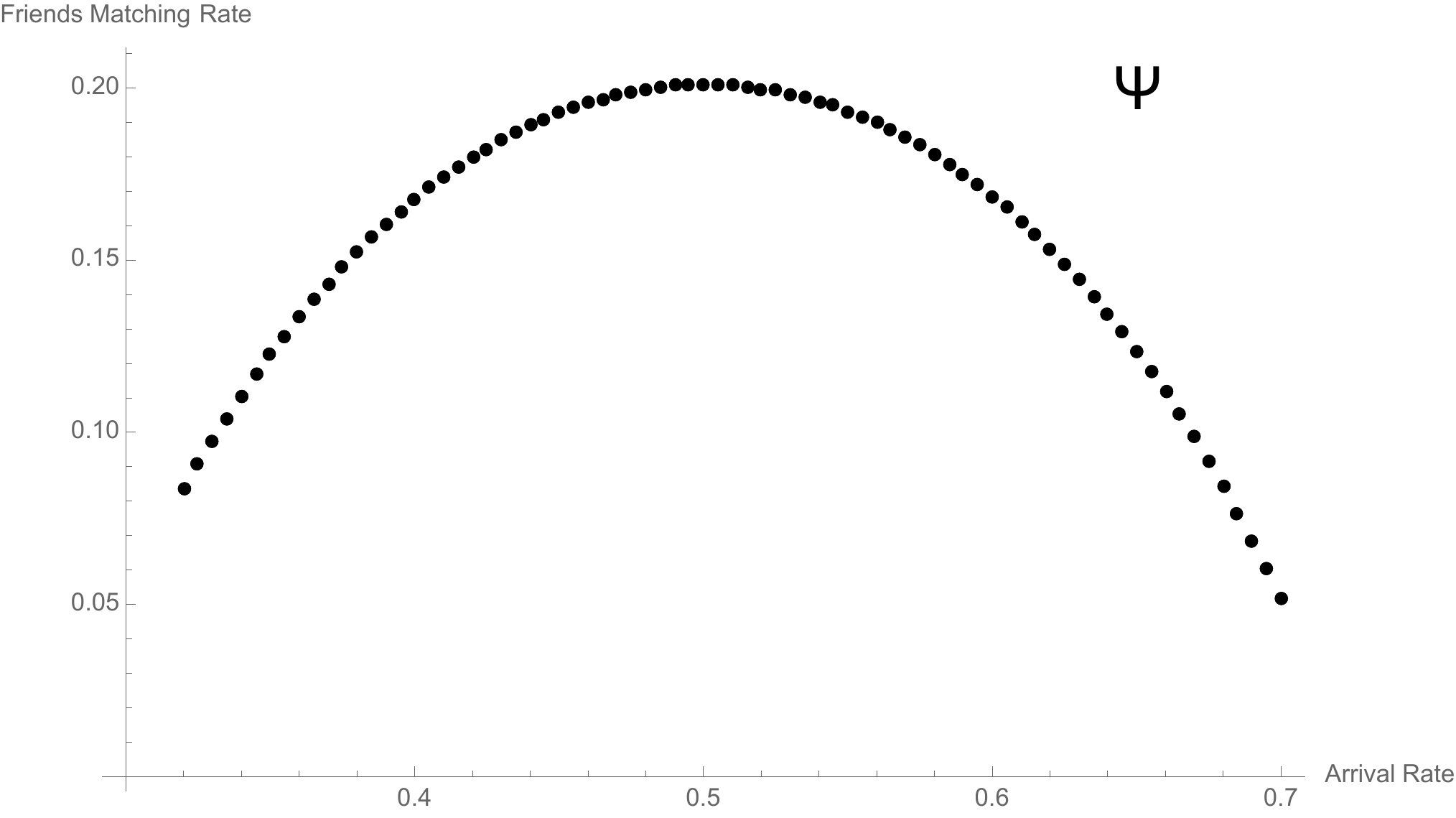}
		\caption{The matching rate through friends $\Psi$ in equilibrium in the model with homogeneous individuals ($h=1$) when $c = .005$, $d = .015$ and $Y = 2$.}
		\label{fig:rates_hom}
	\end{figure}
	
	The possibility of a negative correlation between the network marriage rate and the arrival rate of direct meetings is in sharp contrast from the unambiguously positive correlation that would result in a model where the network is taken as exogenous. Hence, the introduction of endogenous effort in searching for partners via friends \begin{revs}helps rationalizing the empirical patterns documented in Figure \ref{fig:total} as direct meetings via online dating apps became prevalent.\end{revs}
	
	\section{Two levels of education}\label{heterogeneous}
	
	In this section, we consider a marriage market with two levels of education in order to see if our model can replicate the differences in meetings through friends for individuals with and without college education that we presented in Section \ref{facts}.
	
	Let us first point out that also in the model with heterogeneous gains form marriage, the network matching rates are increasing in the arrival rate of direct meetings, $a$, if the social network is taken as exogenous.
	\begin{proposition}\label{prop:cs_het}
		\begin{revs}
			Consider a large marriage market where $h\in (0,1)$ and $s_{i}=s_e$ for all $i\in\mathcal{N}$ of type $e=h,l$. Then, the network matching rates $\Psi(s)_{h,h}$, $\Psi(s)_{l,h}$ and $\Psi(s)_{l,l}$ are increasing in the direct arrival rate of potential partners, $a$.
		\end{revs}
	\end{proposition}
	
	Indeed, Figure \ref{fig:het_exogenous} shows that the friends matching rates for partners met through married friends, denoted by $\Psi$'s, are increasing in the arrival rate of direct offers, $a$, when the network is exogenous, just like in the model with homogeneous individuals. Hence, the introduction of heterogeneity cannot account for the negative correlation between these variables that we reported in Section \ref{facts}.
	
	\begin{figure}
		\centering
		\includegraphics[width=0.7\textwidth]{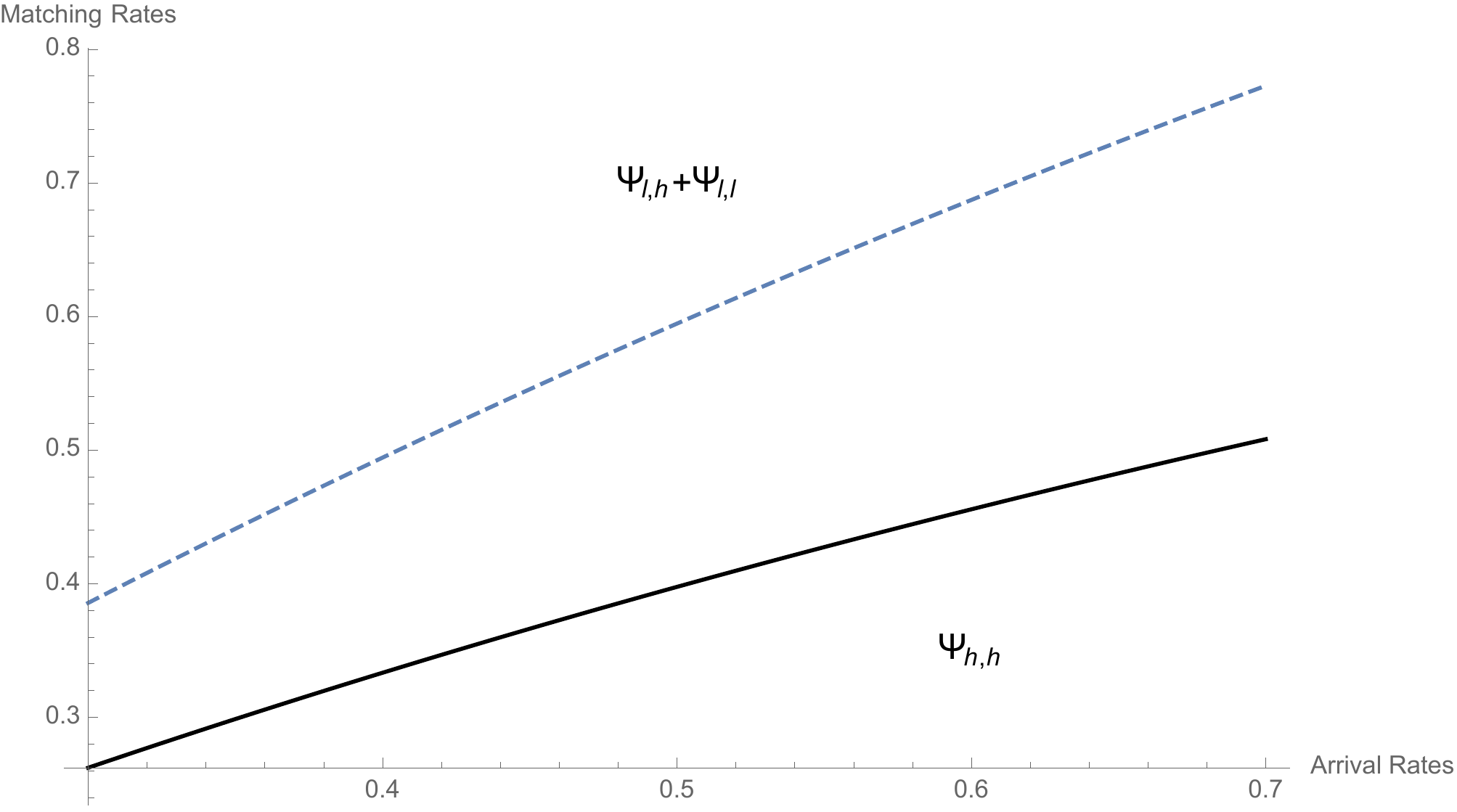}
		\caption{Friends matching rates for $h$ and $l$ individuals in the model with exogenous socialization investments $s_l=s_h=1.5$ when $c = .003$, $d = .015$, $h=.8$ and $Y = 2$.}
		\label{fig:het_exogenous}
	\end{figure}
	
	\begin{revs}
		We then show that an equilibrium of this model with positive investment in socialization exists when socialization costs are low enough.
	\end{revs}
	\begin{proposition}\label{prop:existence}
		\begin{revs}
			Consider a large marriage market where $h\in (0,1)$. Then, there exists $\bar{c}>0$ such that an interior symmetric equilibrium exists for any $c<\bar{c}$.
		\end{revs}
	\end{proposition}
	\begin{revs}
		The proof follows the same logic as in the model with homogeneous individuals, first showing that for low types the marginal utility of the first units of investment in socialization is strictly positive; then, the marginal utility of socialization for high types is always positive, decreasing in $s_l$ and it goes to zero as $s_h$ goes to zero. These facts imply that there must exist values of the socialization cost $c$ that, if low enough, guarantee the existence of an interior equilibrium.\end{revs}\footnote{Note that also in this model an equilibrium with zero socialization levels always exists.}
	
	Introducing heterogeneity complicates the individual's problem. This makes it difficult to derive a formal analysis of the comparative statics of how the investment in socialization and the matching rate change in equilibrium. However, we show by the means of simulations that the model with heterogeneous individuals delivers the same qualitative results than the model with homogeneous individual presented in the previous section.
	
	Figures \ref{fig:socialiation} and \ref{fig:rates} depict how socialization investment and the matching rates through friends for educated and less educated individuals change when the direct arrival rate of potential partners, $a$, changes from $.3$ to $.7$ in a marriage market where we assume that $c =.003$, $d = .015$ and $Y=2$. In Figure \ref{fig:rates}, we report only the meeting rates for partners met through married friends, denoted by $\Psi$'s, as the meeting rates of meetings through a married potential partner met directly who introduced them a friend, denoted by $\Upsilon$'s, are strictly monotonic in socialization investment.
	
	As in the model with homogeneous education levels, we find that both socialization and the matching rates through friends are first increasing and then decreasing in the arrival rate of potential partners. The novel result is that the less educated individuals rely more on friends to look for potential partners than the more educated individuals. These translates into more marriages from friends for the former than for the latter. This result is in line with trends we uncovered in Figures \ref{fig:bundlecol} and \ref{fig:totalcol} looking at how couple met in the last decades in the US.
	
	\begin{revs}
		This result stems from the asymmetric role that marriage through friends has for low and high educated individuals in our model because of the assumption that low skilled individuals do not introduce low skilled individuals to their high skilled friends, i.e., $\Psi_{h,l}=0$. Without this assumption, heterogeneity in education would imply that the first order conditions for both types are equivalent. Hence, the analysis would be as in the model of homogeneous individuals of Section \ref{homogenous}. A higher socialization investment for low skilled individuals could then result because they face a higher risk of divorce (if this is not too high to begin with), in line with the trends in divorce of recent years \citep{harkonen2006}.
	\end{revs}
	
	Additionally, Figures \ref{fig:bundlecol} and \ref{fig:totalcol} also show that college educated people started before to rely less on friends to find a partner. Through the lens of our model, this is an indication that the arrival rate of potential partners for college educated people is higher than for non-college educated people. This is consistent with the idea that the gains from marriage with these partners are higher. It would be interesting in future work to endogenize the arrival rate of direct offers, as in \cite{merlino2019}.
	
	\begin{figure}
		\centering
		\includegraphics[width=0.7\textwidth]{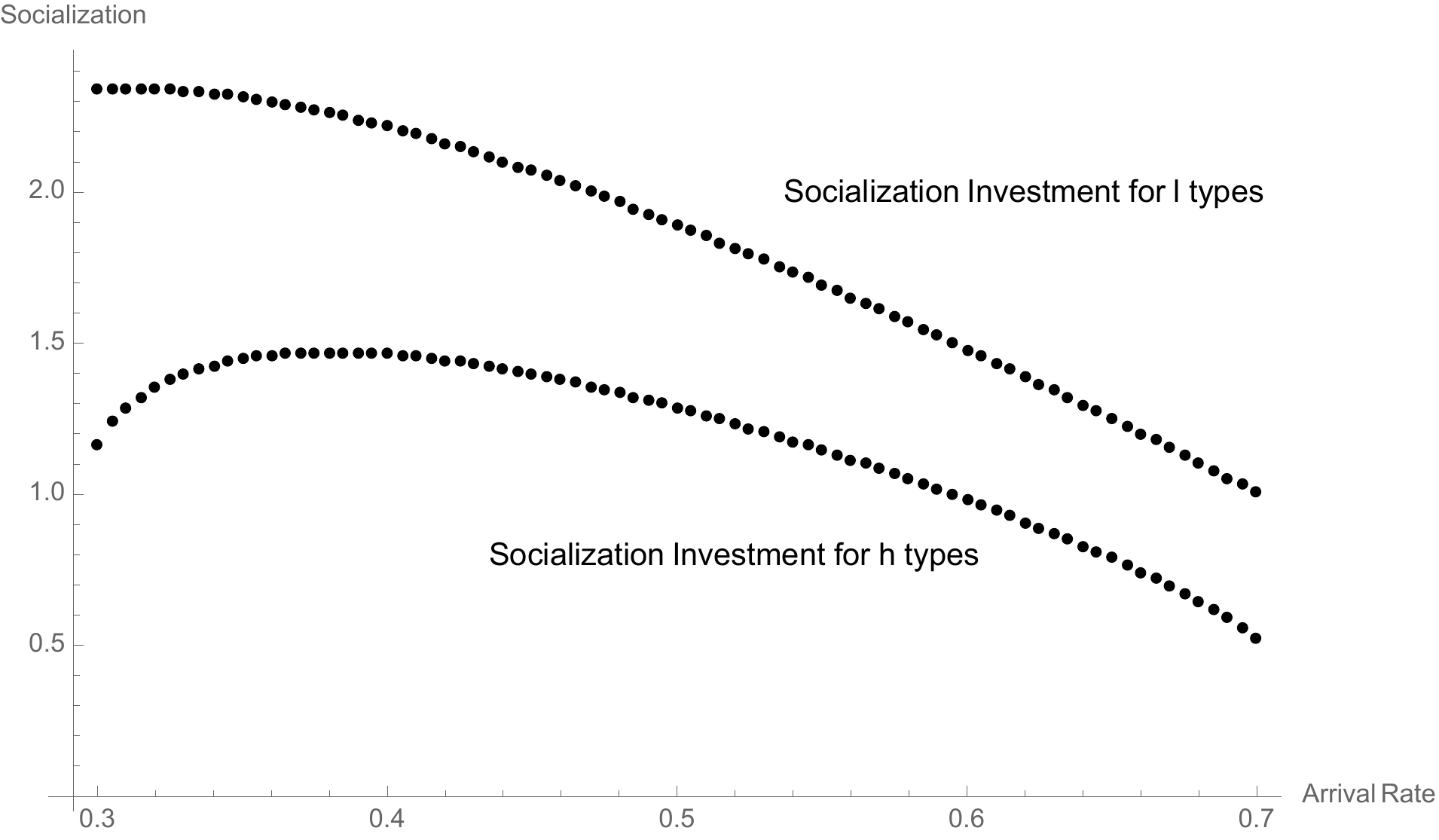}
		\caption{Equilibrium socialization investment for $h$ and $l$ individuals in the model when $c = .003$, $d = .015$, $h=.8$ and $Y = 2$.}
		\label{fig:socialiation}
	\end{figure}
	
	\begin{figure}
		\centering
		\includegraphics[width=0.7\textwidth]{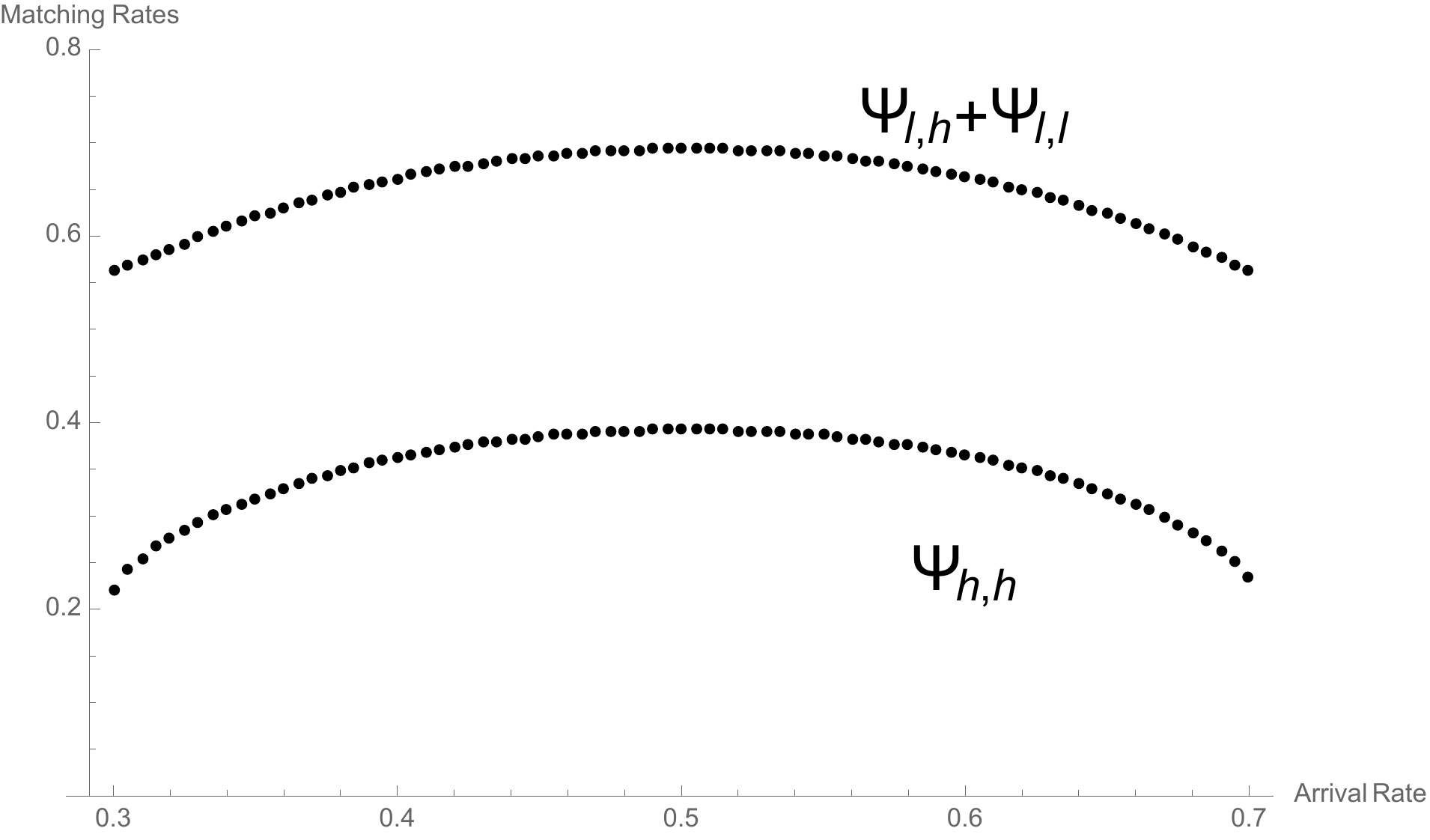}
		\caption{The matching rates through friends $\Psi_{l,h}+\Psi_{l,l}$ for $l$ individuals and the matching rate through friends $\Psi_{h,h}$ for $h$ individuals in equilibrium in the model when $c = .003$, $d = .015$, $h=.8$ and $Y = 2$.}
		\label{fig:rates}
	\end{figure}
	
	\section{Conclusions}\label{conclusions}
	
	In this paper, we propose a model of the heterosexual marriage market in which individuals meet potential partners either directly or through their friends. Individuals invest in socialization in order to make friends, anticipating that their friends can introduce potential partners to them. We show \begin{revs}derive conditions such that an interior exists. Then, we show that, as the arrival rate of potential partners met directly increases,\end{revs} which captures the reduction in search costs that came about with the advent of online dating, socialization investment is first increasing, because there are more potential partners to be met also through friends. When the arrival rate of potential partners met directly increases further, the investment in socialization decreases, as more friends can bring only few additional meetings. The probability of marriages through friends has the same non-monotonic pattern as the investment in socialization. Hence, contrarily to a model in which socialization is exogenous, our model \begin{revs}helps rationalizing\end{revs} the negative correlation between the advent of online dating and the decrease of marriages through friends.
	
	While in this paper we addressed the question of how marriage through friends affects \begin{revs}the marriage rates\end{revs}, future research should understand what are the implications on bargaining within the household. This would be especially relevant when spouses have their own preferences \begin{revs}(see \citealt{chiappori1988,chiappori1992,cherchye2017} among others)\end{revs}, as then one's social network would have an impact on bargaining within the household. Additionally, in this paper we have considered some vertical heterogeneity in spouses' characteristics, such as education. It would be interesting to understand the implication of horizontal dimensions such as race, and in particular the role played by friends in shaping attitudes towards minorities \citep{merlino2019jole}.
	
	\bigskip 
	
	\noindent \textbf{Acknowledgments.} We would like to thank Paolo Pin for comments and suggestions. The usual disclaimers apply.

	\section*{Declarations}
	
	\begin{itemize}
		\item \textbf{Funding.} Financial support from the Research Foundation - Flanders (FWO) through grant G026619N is gratefully acknowledged.
		\item \textbf{Conflict of interest.} The authors have no conflicts of interest to declare that are relevant to the content of this article.
		\item \textbf{Ethics approval.} Not applicable.
		\item \textbf{Consent to participate.} Not applicable.
		\item \textbf{Consent for publication.} Not applicable.
		\item \textbf{Availability of data and materials.} The data used for generating Figures \ref{fig:bundlegen} and \ref{fig:bundlecol} in this article are the ``How Couples Meet and Stay Together 2017'' (HCMST 2017) dataset \citep{dataset2}. These data are freely available to users who register with SSDS/Stanford Libraries. See \url{https://data.stanford.edu/hcmst2017}.
		\item \textbf{Code availability.} The codes used to generate the figures in the paper is available as supplementary material to the article.
		\item \textbf{Authors' contributions.} These authors contributed equally to this work.
	\end{itemize}
	
	\appendix
	
	\setcounter{equation}{0}
	\newcommand\appendixlabel{A}
	\renewcommand{\theequation}{A-\arabic{equation}}
	
	\section*{Appendix A: Proofs}
	
	\begin{proof}[Proof of Proposition \ref{prop:exogenous}]
		To prove the proposition, first note that
		\begin{equation*}
			\frac{\partial \Psi(s,a)}{\partial s}=e^{-\frac{a(1-d)}{d}(1-e^{-s d})}\frac{(1-d)(1-e^{-s d})}{d}>0.
		\end{equation*}
		If individuals are homogeneous and $s_i=s$ for all $i\in N$, the network marriage rate is
		\begin{eqnarray*}
			m(s,a)&=&a(1-d)\Upsilon(s)+(1-a) \Psi(s,a).
		\end{eqnarray*}
		Hence,
		\begin{equation*}\label{psi1}
			\frac{\partial m(s,a)}{\partial s}(s)= a(1-d)d e^{-s d}+(1-a) (1-\Psi(s))\frac{(1-d)(1-e^{-s d})}{d}>0.
		\end{equation*}
		It is trivial to see that $\partial \Psi(s,a)/\partial a >0$, while $\Upsilon(s)$ is independent of $a$. Finally,
		\begin{eqnarray*}
			\frac{\partial m(s,a)}{\partial a}(s)&=&(1-d)\Upsilon(s)-\Psi(s,a)+(1-a) \frac{\partial \Psi(s,a)}{\partial a}=\\
			&=& (1-d)\Upsilon(s)-\Psi(s,a) +(1-a) (1-\Psi(s))\frac{(1-d)(1-e^{-s d})}{d}
		\end{eqnarray*}
		which is null when $a\rightarrow 0$, equal to $(1-s)(1-e^{-sd})$ when $a\rightarrow 1$. Furthermore, $\partial^2 m(s,a)/\partial a$ goes to $(1-d)(1-e^{-sd})(1+1/d)>0$ when $a\rightarrow 0$, and to $-d-(1-d)e^{-sd}+(1-\Psi(s,a))<0$ when $a\rightarrow 1$. So, $m(s,a)$ in first increasing and then decreasing in $a$. This concludes the proof of Proposition \ref{prop:exogenous}.
	\end{proof}
	
	\bigskip
	
	\begin{proof}[Proof of Proposition \ref{prop:NE}]
		Suppose an interior equilibrium exists. Consider a profile $\mathbf{s}$ where $s_{j}=s$, $\forall j\neq i$. Note that $\Upsilon(s_i,s)$ does not depend on $s_i$, while $\Psi(s_i,s)$ can be written as
		\begin{equation*}
			\Psi (s_i,s)= 1-e^{-\frac{a(1-d)h}{d s} s_i \left(1-e^{-d s}\right)}.
		\end{equation*}
		Then, given \eqref{eq:expectutility}, an interior equilibrium $s^*$ solves:
		\begin{equation*}
			\frac{\partial EU_{i}}{\partial s_{i}}(s^*,s^*)=-d(1-a)\frac{\partial \phi _{i}}{\partial s_{i}}(s^*,s^*)-c=0.
		\end{equation*}
		\noindent in large marriage markets, $s^{\ast }$ must solve \eqref{foc_homo}. The LHS of this expression is decreasing in $s^{\ast }$ because both $\left(1-e^{-s^{\ast }d}\right) /s^{\ast }$ and $e^{-\frac{a(1-d)}{d}\left(1-e^{-s^{\ast }d}\right) }$ are decreasing in $s^{\ast }$. Furthermore, when $s^{\ast }$ goes to $0$, the LHS converges to $a d(1-a)(1-d)$, while when $s^{\ast }$ goes to infinity the LHS converges to $0$. Since marginal returns are continuous in $s_{i}$, it follows that an interior symmetric equilibrium exists if and only if $c<a d(1-a)(1-d)$, in which case there is only one symmetric interior equilibrium. This concludes the proof of Proposition \ref{prop:NE}.
	\end{proof}
	
	\bigskip
	
	\begin{proof}[Proof of Proposition \ref{prop:cs}] We first prove part 1. We derive $\partial s^{\ast}/\partial a$ by implicit differentiation of \eqref{foc_homo} and obtain
		\begin{equation}\label{dsda}
			\frac{\partial s^{\ast}}{\partial a}=
			\frac{(1-e^{-s^{\ast}d})\left[1-2 a-a(1-a)(1-d)\frac{1-e^{-s^{\ast}d}}{d} \right]}
			{a(1-a)\left[(a(1-d)e^{-s^*d}+\frac{1}{s})(1-e^{-s^*d})-de^{-s^*d}\right]}.
		\end{equation}
		\noindent First, the denominator is positive. Indeed, its term in the square parenthesis is strictly increasing in $a$ and is equal to $(1-e^{-s^*d})/s-de^{-s^*d}$ when $a=0$. Taking the derivative of this with respect to $d$, we get $sde^{-s^*d}$, which is also strictly positive. Finally, the limit of $(1-e^{-s^*d})/s-de^{-s^*d}$ as $d\rightarrow 0$ is $0$. Hence, the denominator is positive. 
		
		\noindent This implies that the sign of this derivative depends on the sign of the numerator, we see that the term in square parenthesis, which we denote by 
		\begin{equation}\label{num}
			\left(a(1-d)e^{-s^*d}+\frac{1}{s} \right)(1-e^{-s^*d})-de^{-s^*d}
		\end{equation}
		is $1$ when $a=0$, while it is equal to $-1$ when $a=1$. Defining $x=(1-d)\left( 1-e^{-s^{\ast }d}\right)/d$, we rewrite the expression \eqref{num} as $1-2a-a(1-a) x$. While $1-2a-a(1-a) x=0$ admits two solutions, only $(2 + x - \sqrt{4 + x^2})/(2 x)\in[0,1]$. Hence, there is a unique $\bar{a}$ such that if $a\leq\bar{a}$, \eqref{dsda} is positive, while if $a>\bar{a}$, \eqref{dsda} is negative. This concludes the first part of the proof of Proposition \ref{prop:cs}.
		
		\noindent As for the second part of Proposition \ref{prop:cs}, the change in the matching rate $\Psi$ when $a$ changes is described by:
		\begin{equation}\label{der_psi_a}
			\frac{\partial \Psi }{\partial a}=
			(1-\Psi )(1-d)\left[\frac{1-e^{-s^{\ast}d}}{d}+ \frac{\partial s^{\ast}}{\partial a} a e^{-s^{\ast}d}
			\right].
		\end{equation}
		\noindent While the first term in the square parenthesis is always positive, the sign of the second one depends on the sign of $\partial s^{\ast}/\partial a$. Using part 1, $\partial s^{\ast}/\partial a$ is positive when $a$ is sufficiently low, and hence $\partial \Psi/\partial a$ would also be positive. When $a$ tends to $1$, inspecting \eqref{dsda}, we see that the numerator goes to $-(1-e^{-s^{\ast }d})<0$ when $a\rightarrow 1$, while the denominator goes to $0$ from above (as we have just shown that it is always positive). Hence, if $a$ is sufficiently large, $s$ goes to zero, so that the first term in the square brackets of \eqref{der_psi_a} goes to zero, while the second part is negative; hence \eqref{der_psi_a} is negative. 
		
		\noindent The change in the matching rate $\Upsilon$ when $a$ changes is described by:
		\begin{equation}\label{der_upsilon_a}
			\frac{\partial \Upsilon }{\partial a}=
			d(1-d) \left[(1-e^{-s^{\ast}d})+ad e^{-s^{\ast}d} \frac{\partial s^{\ast}}{\partial a} \right].
		\end{equation}
		Following a similar argument as for \eqref{der_psi_a}, we can see that \eqref{der_upsilon_a} is positive for a low enough $a$ and negative for $a$ is large enough.
		
		\noindent Given that $m(s,a)=\Upsilon(s,a)+\Psi(s,a)$, also $m(s,a)$ is increasing in $a$ for a low enough $a$ and decreasing for $a$ is large enough. This concludes the proof of Proposition \ref{prop:cs}.
	\end{proof}

	\begin{proof}[Proof of Proposition \ref{prop:existence}] To prove existence in the socialization effort game in a large society with heterogeneous individuals, we first analyze separately on \textit{low types} and \textit{high types}, respectively. We derive conditions for an interior equilibrium to exists in each scenario, and then we combine these results in order to prove existence in the market with two types. We focus on symmetric equilibria in which $s_i=s_h$ for all $i\in M_H$, and $s_j=s_l$ for all $j\in M_L$.\\
		\textbf{Part 1: Low types}
		Taking the first order conditions of the utility function with respect to $s_i$ for $i\in M_L$, and focusing on the symmetric equilibrium by setting $s_i=s_l$, we obtain the following:
		\begin{eqnarray}
			&\frac{\partial MU_{i,l}}{\partial s_{i,l}}:=-\frac{(a-1) a (d-1)}{(s_l (h
				(s_h-s_l)+s_l))} e^{-\frac{\frac{a (d-1) s_l
						\left(e^{\frac{d (h-1) s_l^2}{h s_h-h
								s_l+s_l}}-1\right)}{s_l}+\frac{a (1-d) h
						s_i \left(1-e^{-d s_h}\right)}{h
						(s_h-s_l)+s_l}+d^2 s_h}{d}} \notag\\	
			&\left[-h s_l (Y-1) e^{\frac{a (d-1) s_l
					\left(e^{\frac{d (h-1) s_l^2}{h s_h-h
							s_l+s_l}}-1\right)}{d s_l}+d
				s_h}\right. +h s_l (Y-1) e^{\frac{a (d-1)
					s_l \left(e^{\frac{d (h-1) s_l^2}{h s_h-h
							s_l+s_l}}-1\right)}{d s_l}}+ \notag\\
			&\left.+(h
			s_h-h s_l+s_l) e^{d \left(\frac{(h-1)
					s_l^2}{h s_h-h
					s_l+s_l}+s_h\right)}-e^{d s_h} (h
			s_h+s_l)+h s_l\right]=c. \label{foc_l}
		\end{eqnarray}
		Taking the limits of $\partial MU_{i,l}/\partial s_{i,l}$ for $s_l\rightarrow 0$ and $s_l\rightarrow \infty$ we obtain the following:
		\begin{eqnarray*}
			\lim\limits_{s_l\rightarrow0} \frac{\partial MU_{i,l}}{\partial s_{i,l}} &=&\frac{(a-1) a (d-1) Y e^{-d s_h} \left(e^{d
					s_h}-1\right)}{s_h},\\
			\lim\limits_{s_l\rightarrow\infty}\frac{\partial MU_{i,l}}{\partial s_{i,l}}&=&0.
		\end{eqnarray*}
		The LHS of equation \eqref{foc_l}, $\partial MU_{i,l}/\partial s_{i,l}$, converges to $(a-1) a (d-1) Y e^{-d s_h} \left(e^{d s_h}-1\right)/s_h$ when taking the limit for $s_l\rightarrow 0$. Note that, if $h\rightarrow 0$, the problem for the low types becomes the same as with one type (see proof of Proposition \ref{prop:NE}). Since the marginal returns are always positive and continuous in $s_l$, a symmetric interior equilibrium for the low types exists if and only if $c<\bar{c}_l=(a-1) a (d-1) Y e^{-d s_h} \left(e^{d s_h}-1\right)/s_h$, which is decreasing in $s_h$.
		\\
		\textbf{Part 2: High types}
		As for individuals of education level $h$, we consider the maximization problem of an individual $i\in M_H$, keeping $s_h$ and $s_l$ fixed. Taking the first order conditions of the utility function with respect to $s_i$, and setting $s_i=s_h$ we obtain the following:
		\begin{equation}
			\frac{\partial MU_{i,h}}{\partial s_{i,h}}:=\frac{(1-a) a (1-d) h Y \left(1-e^{-d s_h}\right)
				e^{-\frac{a (1-d) h s_h \left(1-e^{-d
							s_h}\right)}{d (h
						(s_h-s_l)+s_l)}}}{h
				(s_h-s_l)+s_l}=c. \label{foc_h}
		\end{equation}
		Note that $s_h>0$ and $s_l>0$, LHS of equation \ref{foc_h}, $\partial MU_{i,h}/\partial s_{i,h}$, is always positive and continuous in both $s_h$ and $s_l$. Moreover:
		\begin{eqnarray*}
			\lim\limits_{s_h\rightarrow0}\frac{\partial MU_{i,h}}{\partial s_{h}}&=&0 \\
			\lim\limits_{s_h\rightarrow\infty}\frac{\partial MU_{i,h}}{\partial s_{h}}&=&0
		\end{eqnarray*}
		Hence, we can conclude that there exists a threshold $\bar{c}_h$ such that for any $c<\bar{c}_h$ a solution $s_h$ of $\partial MU_{i,h}/\partial s_{i,h}=c$ in a symmetric equilibrium exists.
		\\
		\textbf{Part 3}: From part 1 of the proof, we have obtained that a symmetric interior equilibrium for low types exists if and only if $c<\bar{c}_l$, and that $\bar{c}_l$ is decreasing in $s_h$. Additionally, from part 2 of the proof, we have obtained that a symmetric interior equilibrium for high types exists if and only if $c<\bar{c}_h$. Since $\partial MU_{i,l}/\partial s_{i,l}$ and $\partial MU_{i,h}/\partial s_{i,h}$ in a symmetric equilibrium are both continuous in $s_h$ and $s_l$, we conclude that there exists a threshold $\bar{c}$ such that for all $c<\bar{c}$ both \eqref{foc_l} and \eqref{foc_l} are simultaneously satisfied, i.e., at least one interior equilibrium exists. This concludes the proof of Proposition \ref{prop:existence}.
	\end{proof}
	
	\begin{proof}[Proof of Proposition \ref{prop:cs_het}]
		We report for convenience the equations for the network matching rates. 
		\begin{eqnarray*}
			\Psi_{h,h}(s) &=& 1-\lim_{n \rightarrow \infty} \phi_{h,h}(s) = 1-e^{-\frac{a(1-d)h s_h}{d (h s_h+(1-h)s_l)}\left(1-e^{-d s_h}\right)}, \\
			\Psi_{l,h}(s)&=&1-e^{-\frac{a(1-d) h s_l}{d (h s_h+(1-h)s_l)}\left(1-e^{-d s_h}\right)}, \\
			\Psi_{l,l}(s)&=&1-e^{-\frac{a(1-d)(1-h)s_l}{d (h s_h+(1-h)s_l)}\left(1-e^{-d (1-h) s_l\frac{sl}{h s_h+(1-h) s_l} }\right)} . 
		\end{eqnarray*}
		Inspecting these expressions, it is immediate to see that, given any positive value of $s_l$ and $s_h$, the network matching rates are increasing in the arrival rate $a$. This concludes the proof of Proposition \ref{prop:cs_het}.
	\end{proof}
	


\begin{thebibliography}{}

\bibitem[\protect\citeauthoryear{Bala and Goyal}{Bala and
  Goyal}{2000}]{balagoyal}
Bala, V. and S.~Goyal (2000).
\newblock A noncooperative model of network formation.
\newblock {\em Econometrica\/}~{\em 68\/}(5), 1181--1229.

\bibitem[\protect\citeauthoryear{Becker}{Becker}{1973}]{becker1973}
Becker, G.~S. (1973).
\newblock A theory of marriage: Part i.
\newblock {\em J Polit Econ\/}~{\em 81\/}(4), 813--846.

\bibitem[\protect\citeauthoryear{Browning, Chiappori, and Weiss}{Browning
  et~al.}{2014}]{browning2014economics}
Browning, M., P.-A. Chiappori, and Y.~Weiss (2014).
\newblock {\em Economics of the Family}.
\newblock Cambridge University Press.

\bibitem[\protect\citeauthoryear{Cabrales, Calv{\'o}-Armengol, and
  Zenou}{Cabrales et~al.}{2011}]{cabrales2011}
Cabrales, A., A.~Calv{\'o}-Armengol, and Y.~Zenou (2011).
\newblock Social interactions and spillovers.
\newblock {\em Games Econ Behav\/}~{\em 72\/}(2), 339--360.

\bibitem[\protect\citeauthoryear{Calv{\'o}-Armengol}{Calv{\'o}-Armengol}{2004}]{calvo2004}
Calv{\'o}-Armengol, A. (2004).
\newblock Job contact networks.
\newblock {\em J Econ Theory\/}~{\em 115\/}(1), 191--206.

\bibitem[\protect\citeauthoryear{Calv{\'o}-Armengol and
  Jackson}{Calv{\'o}-Armengol and Jackson}{2004}]{calvo2004effects}
Calv{\'o}-Armengol, A. and M.~O. Jackson (2004).
\newblock The effects of social networks on employment and inequality.
\newblock {\em Am Econ Rev\/}~{\em 94\/}(3), 426--454.

\bibitem[\protect\citeauthoryear{Chadwick and Solon}{Chadwick and
  Solon}{2002}]{chasol2002}
Chadwick, L. and G.~Solon (2002).
\newblock Intergenerational income mobility among daughters.
\newblock {\em Am Econ Rev\/}~{\em 92\/}(1), 335--344.

\bibitem[\protect\citeauthoryear{Cherchye, Demuynck, De~Rock, and
  Vermeulen}{Cherchye et~al.}{2017}]{cherchye2017}
Cherchye, L., T.~Demuynck, B.~De~Rock, and F.~Vermeulen (2017).
\newblock Household consumption when the marriage is stable.
\newblock {\em Am Econ Rev\/}~{\em 107\/}(6), 1507--34.

\bibitem[\protect\citeauthoryear{Chiappori}{Chiappori}{1988}]{chiappori1988}
Chiappori, P.-A. (1988).
\newblock Rational household labor supply.
\newblock {\em Econometrica\/}~{\em 56\/}(1), 63--90.

\bibitem[\protect\citeauthoryear{Chiappori}{Chiappori}{1992}]{chiappori1992}
Chiappori, P.-A. (1992).
\newblock Collective labor supply and welfare.
\newblock {\em J Polit Econ\/}~{\em 100\/}(3), 437--467.

\bibitem[\protect\citeauthoryear{Chiappori, Salani{\'e}, and Weiss}{Chiappori
  et~al.}{2017}]{chiappori2017}
Chiappori, P.-A., B.~Salani{\'e}, and Y.~Weiss (2017).
\newblock Partner choice, investment in children, and the marital college
  premium.
\newblock {\em Am Econ Rev\/}~{\em 107\/}(8), 2109--67.

\bibitem[\protect\citeauthoryear{Choo and Siow}{Choo and Siow}{2006}]{choo2006}
Choo, E. and A.~Siow (2006).
\newblock Who marries whom and why.
\newblock {\em J Polit Econ\/}~{\em 114\/}(1), 175--201.

\bibitem[\protect\citeauthoryear{Chung and Lu}{Chung and Lu}{2002}]{chung2002}
Chung, F. and L.~Lu (2002).
\newblock The average distances in random graphs with given expected degrees.
\newblock {\em Proc Natl Acad Sci USA\/}~{\em 99\/}(25), 15879--15882.

\bibitem[\protect\citeauthoryear{Currarini, Jackson, and Pin}{Currarini
  et~al.}{2009}]{currarini2009}
Currarini, S., M.~O. Jackson, and P.~Pin (2009).
\newblock An economic model of friendship: Homophily, minorities, and
  segregation.
\newblock {\em Econometrica\/}~{\em 77\/}(4), 1003--1045.

\bibitem[\protect\citeauthoryear{Galenianos}{Galenianos}{2021}]{galenianos2021}
Galenianos, M. (2021).
\newblock Referral networks and inequality.
\newblock {\em Econ J\/}~{\em 131\/}(633), 271--301.

\bibitem[\protect\citeauthoryear{Galeotti and Merlino}{Galeotti and
  Merlino}{2014}]{galeotti2014}
Galeotti, A. and L.~P. Merlino (2014).
\newblock Endogenous job contact networks.
\newblock {\em Int Econ Rev\/}~{\em 55\/}(4), 1201--1226.

\bibitem[\protect\citeauthoryear{Greenwood, Guner, Kocharkov, and
  Santos}{Greenwood et~al.}{2014}]{guner}
Greenwood, J., N.~Guner, G.~Kocharkov, and C.~Santos (2014).
\newblock Marry your like: Assortative mating and income inequality.
\newblock {\em Am Econ Rev\/}~{\em 104\/}(5), 348--53.

\bibitem[\protect\citeauthoryear{Jackson and Wolinsky}{Jackson and
  Wolinsky}{1996}]{jackson96}
Jackson, M.~O. and A.~Wolinsky (1996).
\newblock A strategic model of social and economic networks.
\newblock {\em J Econ Theory\/}~{\em 71\/}(1), 44--74.

\bibitem[\protect\citeauthoryear{Kinateder and Merlino}{Kinateder and
  Merlino}{2017}]{KM}
Kinateder, M. and L.~P. Merlino (2017).
\newblock Public goods in endogenous networks.
\newblock {\em Am Econ J Microecon\/}~{\em 9\/}(3), 187--212.

\bibitem[\protect\citeauthoryear{Kinateder and Merlino}{Kinateder and
  Merlino}{2021}]{KM2}
Kinateder, M. and L.~P. Merlino (2021).
\newblock Free riding in networks.
\newblock Mimeo, available at \url{https://arxiv.org/abs/2110.11651}.

\bibitem[\protect\citeauthoryear{Merlino}{Merlino}{2014}]{merlino2014}
Merlino, L.~P. (2014).
\newblock Formal and informal job search.
\newblock {\em Econ Lett\/}~{\em 125\/}(3), 350--352.

\bibitem[\protect\citeauthoryear{Merlino}{Merlino}{2019}]{merlino2019}
Merlino, L.~P. (2019).
\newblock Informal job search through social networks and vacancy creation.
\newblock {\em Econ Lett\/}~{\em 178}, 82--85.

\bibitem[\protect\citeauthoryear{Merlino, Steinhardt, and Wren-Lewis}{Merlino
  et~al.}{2019}]{merlino2019jole}
Merlino, L.~P., M.~F. Steinhardt, and L.~Wren-Lewis (2019).
\newblock More than just friends? {School} peers and adult interracial
  relationships.
\newblock {\em J Labor Econ\/}~{\em 37\/}(3), 663--713.

\bibitem[\protect\citeauthoryear{Mourifi{\'e} and Siow}{Mourifi{\'e} and
  Siow}{2021}]{siow2021}
Mourifi{\'e}, I. and A.~Siow (2021).
\newblock The {Cobb-Douglas} marriage matching function: Marriage matching with
  peer and scale effects.
\newblock {\em J Labor Econ\/}~{\em 39\/}(S1), S239--S274.

\bibitem[\protect\citeauthoryear{Pencavel}{Pencavel}{1998}]{Pen1998}
Pencavel, J. (1998).
\newblock Assortative mating by schooling and the work behavior of wives and
  husbands.
\newblock {\em Am Econ Rev\/}~{\em 88\/}(2), 326--329.

\bibitem[\protect\citeauthoryear{Rosenfeld, Reuben, and Hausen}{Rosenfeld
  et~al.}{2019}]{dataset2}
Rosenfeld, M.~J., J.~T. Reuben, and S.~Hausen (2019).
\newblock How couples meet and stay together 2017 fresh sample. [computer
  files].
\newblock data retrieved from \url{https://data.stanford.edu/hcmst2017}.

\bibitem[\protect\citeauthoryear{Rosenfeld, Thomas, and Hausen}{Rosenfeld
  et~al.}{2019}]{rosenfeld2019}
Rosenfeld, M.~J., R.~J. Thomas, and S.~Hausen (2019).
\newblock Disintermediating your friends: How online dating in the {United
  States} displaces other ways of meeting.
\newblock {\em Proc Natl Acad Sci USA\/}~{\em 116\/}(36), 17753--17758.

\bibitem[\protect\citeauthoryear{Schwartz}{Schwartz}{2013}]{schwartz2013trends}
Schwartz, C.~R. (2013).
\newblock Trends and variation in assortative mating: Causes and consequences.
\newblock {\em Annu Rev Sociol\/}~{\em 39}, 451--470.

\bibitem[\protect\citeauthoryear{Schwartz and Mare}{Schwartz and
  Mare}{2005}]{schwartz2005trends}
Schwartz, C.~R. and R.~D. Mare (2005).
\newblock Trends in educational assortative marriage from 1940 to 2003.
\newblock {\em Demography\/}~{\em 42\/}(4), 621--646.

\end{thebibliography}


\begin{thebibliography}{}
		
		\bibitem[\protect\citeauthoryear{Bala and Goyal}{Bala and
			Goyal}{2000}]{balagoyal}
		Bala, V. and S.~Goyal (2000).
		\newblock A noncooperative model of network formation.
		\newblock {\em Econometrica\/}~{\em 68\/}(5), 1181--1229.
		
		\bibitem[\protect\citeauthoryear{Becker}{Becker}{1973}]{becker1973}
		Becker, G.~S. (1973).
		\newblock A theory of marriage: Part I.
		\newblock {\em J Polit Econ\/}~{\em 81\/}(4), 813--846.
		
		\bibitem[\protect\citeauthoryear{Bellou}{Bellou}{2015}]{bellou2015}
		Bellou, A. (2015).
		\newblock The impact of internet diffusion on marriage rates: Evidence from the broadband market.
		\newblock {\em J Pop Econ\/}~{\em 28\/}(2), 265--297.
		
		\bibitem[\protect\citeauthoryear{Browning, Chiappori, and Weiss}{Browning et~al.}{2014}]{browning2014economics}
		Browning, M., P.-A. Chiappori, and Y.~Weiss (2014).
		\newblock {\em Economics of the Family}.
		\newblock Cambridge University Press.
		
		\bibitem[\protect\citeauthoryear{Cabrales, Calv{\'o}-Armengol, and Zenou}{Cabrales et~al.}{2011}]{cabrales2011}
		Cabrales, A., A.~Calv{\'o}-Armengol, and Y.~Zenou (2011).
		\newblock Social interactions and spillovers.
		\newblock {\em Games Econ Behav\/}~{\em 72\/}(2), 339--360.
		
		\bibitem[\protect\citeauthoryear{Calv{\'o}-Armengol}{Calv{\'o}-Armengol}{2004}]{calvo2004}
		Calv{\'o}-Armengol, A. (2004).
		\newblock Job contact networks.
		\newblock {\em J Econ Theory\/}~{\em 115\/}(1), 191--206.
		
		\bibitem[\protect\citeauthoryear{Calv{\'o}-Armengol and Jackson}{Calv{\'o}-Armengol and Jackson}{2004}]{calvo2004effects}
		Calv{\'o}-Armengol, A. and M.~O. Jackson (2004).
		\newblock The effects of social networks on employment and inequality.
		\newblock {\em Am Econ Rev\/}~{\em 94\/}(3), 426--454.
		
		
		\bibitem[\protect\citeauthoryear{Cherchye, Demuynck, De~Rock, and Vermeulen}{Cherchye et~al.}{2017}]{cherchye2017}
		Cherchye, L., T.~Demuynck, B.~De~Rock, and F.~Vermeulen (2017).
		\newblock Household consumption when the marriage is stable.
		\newblock {\em Am Econ Rev\/}~{\em 107\/}(6), 1507--34.
		
		\bibitem[\protect\citeauthoryear{Chiappori}{Chiappori}{1988}]{chiappori1988}
		Chiappori, P.~A. (1988).
		\newblock Rational household labor supply.
		\newblock {\em Econometrica\/}~{\em 56\/}(1), 63--90.
		
		\bibitem[\protect\citeauthoryear{Chiappori}{Chiappori}{1992}]{chiappori1992}
		Chiappori, P.~A. (1992).
		\newblock Collective labor supply and welfare.
		\newblock {\em J Polit Econ\/}~{\em 100\/}(3), 437--467.
		
		\bibitem[\protect\citeauthoryear{Chiappori, Salani{\'e}, and Weiss}{Chiappori et~al.}{2017}]{chiappori2017}
		Chiappori, P.~A., B.~Salani{\'e}, and Y.~Weiss (2017).
		\newblock Partner choice, investment in children, and the marital college premium.
		\newblock {\em Am Econ Rev\/}~{\em 107\/}(8), 2109--67.
		
		\bibitem[\protect\citeauthoryear{Choo and Siow}{Choo and Siow}{2006}]{choo2006}
		Choo, E. and A.~Siow (2006).
		\newblock Who marries whom and why.
		\newblock {\em J Polit Econ\/}~{\em 114\/}(1), 175--201.
		
		
		\bibitem[\protect\citeauthoryear{Currarini, Jackson, and Pin}{Currarini et~al.}{2009}]{currarini2009}
		Currarini, S., M.~O. Jackson, and P.~Pin (2009).
		\newblock An economic model of friendship: Homophily, minorities, and segregation.
		\newblock {\em Econometrica\/}~{\em 77\/}(4), 1003--1045.
		
		\bibitem[\protect\citeauthoryear{Galenianos}{Galenianos}{2021}]{galenianos2021}
		Galenianos, M. (2021).
		\newblock Referral networks and inequality.
		\newblock {\em Econ J\/}~{\em 131\/}(633), 271--301.
		
		\bibitem[\protect\citeauthoryear{Galeotti and Merlino}{Galeotti and Merlino}{2014}]{galeotti2014}
		Galeotti, A. and L.~P. Merlino (2014).
		\newblock Endogenous job contact networks.
		\newblock {\em Int Econ Rev\/}~{\em 55\/}(4), 1201--1226.
		
		\bibitem[\protect\citeauthoryear{Gould and Paserman}{Gould and Paserman}{2003}]{gould2003}
		Gould ED, Paserman MD (2003).
		\newblock Waiting for Mr. Right: Rising inequality and declining marriage rates.
		\newblock {\em J Urban Econ\/}~{ 53\/}(2), 257--281
		
		\bibitem[\protect\citeauthoryear{Greenwood, Guner, Kocharkov, and Santos}{Greenwood et~al.}{2014}]{guner}
		Greenwood, J., N.~Guner, G.~Kocharkov, and C.~Santos (2014).
		\newblock Marry your like: Assortative mating and income inequality.
		\newblock {\em Am Econ Rev\/}~{\em 104\/}(5), 348--53.
		
		
		\bibitem[\protect\citeauthoryear{H{\"a}rk{\"o}nen and Dronkers}{H{\"a}rk{\"o}nen and Dronkers}{2006}]{harkonen2006}
		H{\"a}rk{\"o}nen J, Dronkers J (2006).
		\newblock Stability and change in the educational gradient of divorce. A comparison of seventeen countries.
		\newblock {\em Eur Sociol Rev\/}~{\em 22\/}(5), 501--517.
		
		\bibitem[\protect\citeauthoryear{Jackson and Wolinsky}{Jackson and
			Wolinsky}{1996}]{jackson96}
		Jackson, M.~O. and A.~Wolinsky (1996).
		\newblock A strategic model of social and economic networks.
		\newblock {\em J Econ Theory\/}~{\em 71\/}(1), 44--74.
		
		\bibitem[\protect\citeauthoryear{Kinateder and Merlino}{Kinateder and
			Merlino}{2017}]{KM}
		Kinateder, M. and L.~P. Merlino (2017).
		\newblock Public goods in endogenous networks.
		\newblock {\em Am Econ J Microecon\/}~{\em 9\/}(3), 187--212.
		
		\bibitem[\protect\citeauthoryear{Kinateder and Merlino}{Kinateder and
			Merlino}{2021}]{KM2}
		Kinateder, M. and L.~P. Merlino (2021).
		\newblock Free riding in networks.
		\newblock Mimeo, available at \url{https://arxiv.org/abs/2110.11651}.
		
		\bibitem[\protect\citeauthoryear{Kinateder and Merlino}{Kinateder and
			Merlino}{2022}]{KM3}
		Kinateder, M. and L.~P. Merlino (2022).
		\newblock Local public goods with weighted link formation.
		\newblock {\em Games Econ Behav\/}~{\em 132\/}, 316--327.
		
		\bibitem[\protect\citeauthoryear{Merlino}{Merlino}{2014}]{merlino2014}
		Merlino, L.~P. (2014).
		\newblock Formal and informal job search.
		\newblock {\em Econ Lett\/}~{\em 125\/}(3), 350--352.
		
		\bibitem[\protect\citeauthoryear{Merlino}{Merlino}{2019}]{merlino2019}
		Merlino, L.~P. (2019).
		\newblock Informal job search through social networks and vacancy creation.
		\newblock {\em Econ Lett\/}~{\em 178}, 82--85.
		
		\bibitem[\protect\citeauthoryear{Merlino, Steinhardt, and Wren-Lewis}{Merlino
			et~al.}{2019}]{merlino2019jole}
		Merlino, L.~P., M.~F. Steinhardt, and L.~Wren-Lewis (2019).
		\newblock More than just friends? {School} peers and adult interracial relationships.
		\newblock {\em J Labor Econ\/}~{\em 37\/}(3), 663--713.
		
		\bibitem[\protect\citeauthoryear{Mourifi{\'e} and Siow}{Mourifi{\'e} and
			Siow}{2021}]{siow2021}
		Mourifi{\'e}, I. and A.~Siow (2021).
		\newblock The {Cobb-Douglas} marriage matching function: Marriage matching with peer and scale effects.
		\newblock {\em J Labor Econ\/}~{\em 39\/}(S1), S239--S274.
		
		
		\bibitem[\protect\citeauthoryear{Rosenfeld, Reuben, and Hausen}{Rosenfeld
			et~al.}{2019}]{dataset2}
		Rosenfeld, M.~J., J.~T. Reuben, and S.~Hausen (2019).
		\newblock How couples meet and stay together 2017 fresh sample. [computer files].
		\newblock Data retrieved from \url{https://data.stanford.edu/hcmst2017}.
		
		\bibitem[\protect\citeauthoryear{Rosenfeld and Thomas}{Rosenfeld and Thomas}{2012}]{rosenfeld2012}
		Rosenfeld M.~J. and R.~J. Thomas (2012).
		\newblock Searching for a mate: The rise of the internet as a social intermediary.
		\newblock {\em Am Sociol Rev\/}~{\em 77\/}(4), 523--547.
		
		\bibitem[\protect\citeauthoryear{Rosenfeld, Thomas, and Hausen}{Rosenfeld
			et~al.}{2019}]{rosenfeld2019}
		Rosenfeld, M.~J., R.~J. Thomas, and S.~Hausen (2019).
		\newblock Disintermediating your friends: How online dating in the {United
			States} displaces other ways of meeting.
		\newblock {\em Proc Natl Acad Sci USA\/}~{\em 116\/}(36), 17753--17758.
		
		\bibitem[\protect\citeauthoryear{Schwartz}{Schwartz}{2013}]{schwartz2013trends}
		Schwartz, C.~R. (2013).
		\newblock Trends and variation in assortative mating: Causes and consequences.
		\newblock {\em Annu Rev Sociol\/}~{\em 39}, 451--470.
		
		\bibitem[\protect\citeauthoryear{Schwartz and Mare}{Schwartz and
			Mare}{2005}]{schwartz2005trends}
		Schwartz, C.~R. and R.~D. Mare (2005).
		\newblock Trends in educational assortative marriage from 1940 to 2003.
		\newblock {\em Demography\/}~{\em 42\/}(4), 621--646.
		
		\bibitem[\protect\citeauthoryear{Skopek, Schulz, and Blossfeld}{Skopek et~al.}{2010}]{skopek}
		Skopek J, Schulz F, Blossfeld HP (2010). 
		\newblock Who Contacts Whom? Educational
		Homophily in Online Mate Selection.
		\newblock {\em Eur Sociol Rev\/}~{\em 27\/}(2), 180--195.
		
		\bibitem[\protect\citeauthoryear{Soetevent and Kooreman}{Soetevent and Kooreman}{2007}]{soetevent}
		Soetevent A.~R., Kooreman P. (2007).
		\newblock A discrete-choice model with social interactions: With an application to high school teen behavior.
		\newblock {\em J Appl Econom\/}~{\em 22\/}(3), 599--624.
		
		\bibitem[\protect\citeauthoryear{Stevenson and Wolfers}{Stevenson and Wolfers}{2007}]{stevenson2007}
		Stevenson B, Wolfers J (2007)
		\newblock Marriage and divorce: Changes and their driving forces.
		\newblock {\em J Econ Perspect\/}~{\em 21\/}(2), 27--52
		
		
	\end{thebibliography}
\end{document}